\documentclass[11pt,letterpaper]{article}

\usepackage[margin=1in]{geometry}
\usepackage{hyperref}
\usepackage{algorithm}
\usepackage{algorithmic}

\usepackage{hyperref,enumitem}
\hypersetup{
	colorlinks=true,
	citecolor = blue       %
}
\setlist{noitemsep,leftmargin=\parindent,topsep=2pt}
\setlist{noitemsep,topsep=2pt}
\usepackage{natbib}
\usepackage{url}            %
\usepackage{xspace}
\usepackage{graphicx}
\usepackage{amsmath,amssymb,amsthm}
\usepackage{booktabs} %
\usepackage{color}
\usepackage[font=small,labelfont=bf]{caption}
\usepackage{subcaption}
\usepackage{multirow}
\usepackage{authblk}

\usepackage{comment}

\usepackage{tikz}
\usetikzlibrary{arrows}
\usetikzlibrary{shapes}
\usetikzlibrary{calc}

\newcount\Comments 
\newcommand{\kibitz}[2]{\ifnum\Comments=1{\color{#1}{#2}}\fi}

\newcommand{\E}{\mathbb{E}}
\DeclareMathOperator*{\argmin}{arg\,min}

\newcommand{\eps}{\varepsilon}

\newcommand{\M}{\mathcal{M}}
\newcommand{\R}{\mathbb{R}}
\newcommand{\Rev}{\mathtt{Rev}}

\renewcommand{\O}{\mathcal{O}}
\newcommand{\T}{\mathcal{T}}
\newcommand{\F}{\mathcal{F}}

\newcommand{\PP}{\mathbb{P}}

\newcommand{\C}{\mathcal{C}}

\theoremstyle{plain}
\newtheorem{theorem}{Theorem}

\newtheorem{lemma}{Lemma}
\newtheorem{definition}{Definition}
\newtheorem{assumption}{Assumption}
\newtheorem{main_theorem}{Main Result}

\newtheorem{example}{Example}
\newtheorem{claim}{Claim}

\newcommand{\newnew}[1]{\kibitz{black}{#1}}

\definecolor{ao(english)}{rgb}{0.0, 0.5, 0.0}

\newcount\CommentsAdd 
\newcommand{\kibitzAdd}[2]{\ifnum\CommentsAdd=1{\color{#1}{#2}}\fi}
\newcommand{\dcpadd}[1]{\kibitzAdd{black}{#1}}

\begin{document}

\title{Welfare-Preserving $\varepsilon$-BIC to BIC Transformation with\\ Negligible Revenue Loss\thanks{Part of the work was done when Zhe Feng was a PhD student at Harvard University, where he was supported by a Google PhD fellowship. The first version of this paper was posted on arXiv on July 19, 2020.}
}

\author[a]{Vincent Conitzer}
\author[b]{Zhe Feng}
\author[c]{David C. Parkes}
\author[d]{Eric Sodomka}

\affil[a]{Duke University \authorcr \texttt{conitzer@cs.duke.edu}}
\affil[b]{Google Research \authorcr \texttt{zhef@google.com}}
\affil[c]{Harvard University \authorcr \texttt{parkes@eecs.harvard.edu}}
\affil[d]{Facebook Research \authorcr \texttt{sodomka@facebook.com}}

\date{August 3, 2021}
\maketitle

\begin{abstract}
In this paper, we provide a transform from an $\varepsilon$-BIC
mechanism into an exactly BIC mechanism without any loss of social
welfare and with additive and negligible revenue loss.  This is the
first $\varepsilon$-BIC to BIC transformation that preserves welfare
and provides negligible revenue loss.  The revenue loss bound is
tight given the requirement to maintain social welfare.  Previous
$\varepsilon$-BIC to BIC transformations preserve social welfare but
have no revenue guarantee~\citep{BeiHuang11}, or suffer welfare loss
while incurring a revenue loss with both a multiplicative and an
additive term, e.g.,~\citet{DasWeinberg12, Rubinstein18,
Cai19}. The revenue loss achieved by our transformation is
incomparable to these earlier approaches and can be significantly
less. \newnew{Our approach is different from the previous replica-surrogate matching methods and we directly make use of a directed and weighted type graph (induced by the types' regret), one for each agent. The transformation runs a \emph{fractional rotation step} and a \emph{payment reducing step} iteratively to make the mechanism Bayesian incentive compatible.} We also analyze $\varepsilon$-expected ex-post IC
($\varepsilon$-EEIC) mechanisms~\citep{DuettingFJLLP12}. We provide
a welfare-preserving transformation in this setting with the same
revenue loss guarantee for uniform type distributions and give an
impossibility result for non-uniform distributions.  We apply the
transform to linear-programming based and machine-learning based
methods of automated mechanism design.
\end{abstract}

\thispagestyle{empty}
\clearpage

\section{Introduction}

Optimal mechanism design is very challenging in multi-dimensional settings such as those for selling multiple items, \dcpadd{such as those that arise in the sale of wireless spectrum licenses or the allocation of advertisements to slots in internet advertising. Recognizing this challenge, } there is considerable interest in adopting algorithmic approaches to address these problems of economic design. These include polynomial-time black-box reductions from multi-dimensional revenue maximization to the algorithmic problem for virtual welfare optimization,~e.g.,\citep{CaiDW12a, CaiDW12b, CaiDW13}, and the application of methods from linear programming~\citep{ConitzerS02,ConitzerS04a} and machine learning~\citep{DuettingFJLLP12, FengEtAl18, Duetting19} to automated mechanism design.

 %
%

%
%

\dcpadd{ Moreover, it is common in practical settings that it is important consider both social welfare (efficiency) and revenue.  For example, national governments that use auctions to sell wireless spectrum licenses care both about the efficiency of the allocation as this promotes valuable use as well as the revenue that flows from auctions into the budget.}
In regard to online advertising, there are various works that explore this trade-off between welfare and revenue. Display advertising has focused on yield optimization (i.e., maximizing a combination of revenue and the quality of ads shown)~\citep{Balseiro17}, and work in sponsored search auctions has considered a squashing parameter that trades off efficiency and revenue~\citep{LahaiePe07}. At the same time, there is a surprisingly small theoretical literature that considers both welfare and revenue properties together (e.g., \citet{DPPS12}).

\dcpadd{At the same time, the use of computational methods for economic design often comes with} a limitation, which is that the output mechanism may only be approximately incentive compatible (IC); e.g., the black-box reductions are approximately IC when the algorithmic problems are solved in polynomial time, the LP approach works on a discretized space to reduce computational cost but thereby achieves a mechanism that is only approximately IC in the full space, and the machine learning approaches train a mechanism over finite training data and achieve approximate IC on the full type distribution.
 While  it has been debated as to whether approximate incentive compatibility may suffice, e.g.,~\citep{Carroll12,LubinParkes12, AzeyedoBudish19}, %
 this does add an additional layer of unpredictability to the performance of a designed
 mechanism. First, the fact that an agent can gain only a small amount from deviating does not preclude strategic behavior---perhaps the agent can easily  identify a useful deviation, for example through repeated interactions, that reliably provides increased profit.
 This can be a problem when strategic responses lead to an unraveling of the desired economic properties of the mechanism (we provide such an example in this paper). 
 The possibility of strategic reports by participants has additional consequences as well, for example making it more challenging for a designer to  confidently measure ex-post welfare after outcomes are realized.

 For the above reasons, there is considerable interest in methods to
 transform an $\eps$-Bayesian incentive compatible ($\eps$-BIC) mechanism to an exactly BIC mechanism~\citep{DasWeinberg12, Rubinstein18, Cai19}, or  an $\eps$-expected ex-post IC ($\eps$-EEIC) mechanisms~\citep{DuettingFJLLP12,Duetting19} into
an exactly BIC mechanism.  The main question we want to answer in this paper is:
\begin{quote}
\emph{Given an $\eps$-BIC/$\eps$-EEIC mechanism, is there an exact BIC mechanism that maintains social welfare and achieves negligible revenue loss compared with the original mechanism \dcpadd{under truthful reports}? If so, can we find the transformed, BIC mechanism efficiently?}
\end{quote}

\dcpadd{In this paper, we provide the first $\eps$-BIC to BIC transform that is welfare-preserving while also ensuring only negligible revenue loss relative to the baseline mechanism. This simultaneous attention to the properties of both welfare and revenue is of practical importance. An immediate corollary of our main result is the well known result from economic theory, namely that efficient allocations can be implemented in an incentive-compatible way. For example, the transform can be applied to a first-price, sealed-bid auction to achieve an efficient and BIC auction.

Our approach is different from the previous replica-surrogate matching methods and we directly make use of a directed and weighted type graph (induced by the types' regret), one for each agent. The transformation runs a \emph{fractional rotation step} and a \emph{payment reducing step} iteratively to make the mechanism Bayesian incentive compatible.

The transform also satisfies
another appealing property, which is that of {\em allocation-invariance}. The transformed mechanism maintains the same distribution on outcomes, allocations for example, as the baseline mechanism (focusing here on the non-monetary part of the output of the mechanism).\footnote{This allocation-invariance is \emph{ex ante}, i.e., it is with respect to the prior distribution over types.} This property is useful in many scenarios. Consider, for example, a principal such as Amazon that is running a market and also incurs a resource cost for different outcomes (e.g. warehouse storage cost). With this allocation-invariance property, then not only is the welfare the same (or better) and the revenue loss negligible, but the resource cost (averaged over iterations of the mechanism) of the principal is preserved by the transform. }

\subsection{Model and Notation}

We consider a general mechanism design setting with a set of $n$ agents $N = \{1, \ldots, n\}$. Each agent $i$ has a private type $t_i$. We denote the entire type profile as $t = (t_1, \ldots, t_n)$, which is drawn from a joint distribution $\F$. Let $\F_i$ be the marginal distribution of agent $i$ and $\T_i$ be the support of $\F_i$. Let $t_{-i}$ be the joint type profile of the other agents, $\F_{-i}$ be the associated marginal type distribution. Let $\T = \T_1\times\cdots\times\T_n$ and $\T_{-i}$ be the support of $\F$ and $\F_{-i}$, respectively.  In this setting, there is a set of feasible {\em outcomes} denoted by $\O$, typically an allocation of items to agents. Later in the paper, we sometimes also use ``outcome'' to refer to the output of the mechanism, namely the allocation together with the  payments, when this is clear from the context.

We focus on the discrete type setting, i.e., $\T_i$ is a finite set containing $m_i$ possible types, i.e., $|\T_i| = m_i$. Let $t_{i}^{(j)}$ denote the $j$th possible type of agent $i$, where $j\in [m_i]$. For all $i$ and $t_i$, $v_i: (t_i, o) \rightarrow \R_{\geq 0}$ is a valuation that maps a type $t_i$ and outcome $o$ to a non-negative real number. A \emph{direct revelation} mechanism $\M = (x, p)$ is a pair of {\em allocation rule} $x_i: \T \rightarrow \Delta(\O)$, possibly randomized, and {\em expected payment rule} $p_i: \T \rightarrow \R_{\geq 0}$.
We slightly abuse notation, and also use $v_i$ to define the expected value of bidder $i$ for mechanism $\M$, with the expectation taken with respect to the randomization used by the mechanism, that is 
\begin{eqnarray} \forall i, \hat{t}\in \T, v_i(t_i, x(\hat{t})) = \E_{o\sim x(\hat{t})} [v_i(t_i, o)],  
\end{eqnarray}
for true type $t_i$ and reported type profile $\hat{t}$. %
When the reported types are  $\hat{t}=(\hat{t}_1,\ldots,\hat{t}_n)$, the output of mechanism $\M$ for agent $i$ is denoted as $\M_i(\hat{t}) = (x_i(\hat{t}), p_i(\hat{t}))$. We define the {\em utility} of agent $i$ with true type $t_i$ and a reported type $\hat{t}_i$ given the reported type profile $\hat{t}_{-i}$
of other agents as a quasilinear function, 
\begin{eqnarray}
u_i(t_i, \M(\hat{t})) = v_i(t_i, x(\hat{t})) - p_i(\hat{t}).
\end{eqnarray} 

For a multi-agent setting, it will be useful to also define the  {\em interim rules}.
\begin{definition}[Interim Rules of a Mechanism]\label{def:interim-rule}
For a mechanism $\M$ with allocation rule $x$ and payment rule $p$, the interim allocation rule $X$ and payment rule $P$ are defined as, $\forall i, t_i\in\T_i, X_i(t_i) = \E_{t_{-i}\in \F_{-i}}[x_i(t_i; t_{-i})], P_i(t_i) = \E_{t_{-i}\in \F_{-i}}[p_i(t_i; t_{-i})]$.
\end{definition}
In this paper, we assume we have oracle access to the interim quantities of mechanism $\M$.
\begin{assumption}[Oracle Access to Interim Quantities]\label{asmp:interim-oracle}
For any mechanism $\M$, given any type profile $t=(t_1, \ldots, t_n)$, we receive the interim allocation rule $X_i(t_i)$ and payments $P_i(t_i)$, for all $i, t_i$.
\end{assumption}

We define the menu of a mechanism $\M$ in the following way.
\begin{definition}[Menu]\label{def:menu}
For a mechanism $\M$, the menu of bidder $i$ is the set $\{\M_i(t)\}_{t\in \T}$. The menu size of agent $i$ is denoted as $|\M_i|$.
\end{definition}

In mechanism design, there is a focus on designing {\em incentive compatible} mechanisms, so that truthful reporting of types is an equilbrium. This is without loss of generality by the
revelation principle.

It has also been useful to work with approximate-IC mechanisms, and
these have been studied in various papers, e.g.~\citep{DasWeinberg12, CaiZhao17, Rubinstein18, Cai19, DuettingFJLLP12, Duetting19, FengEtAl18, BSV19, LahaieMeSi18, Fengss19}.

In this paper, we focus on two definitions of approximate incentive compatibility, $\eps$-BIC and $\eps$-expected ex post incentive compatible ($\eps$-EEIC). %
\begin{definition}[$\eps$-BIC Mechanism]\label{def:eps-BIC} A mechanism $\M$ is $\eps$-BIC iff for all $i, t_i$,
\begin{eqnarray*}
\E_{t_{-i}\sim \F_{-i}}[u_i(t_i, \M(t))] \geq \max_{\hat{t}_i \in \T_i}\E_{t_{-i}\sim \F_{-i}}[u_i(t_i, \M(\hat{t}_i; t_{-i}))] - \eps
\end{eqnarray*}
\end{definition}

\begin{definition}[$\varepsilon$-expected ex post IC ($\varepsilon$-EEIC) Mechanism~\citep{DuettingFJLLP12}]\label{def:eps-EEIC}
A mechanism $\M$ is $\eps$-EEIC if and only if for all $i$, $\E_t\left[\max_{\hat{t}_i \in \T_i} u_i(t_i, \M(\hat{t}_i; t_{-i})) - u_i(t_i, \M(t))\right]\leq \epsilon$.
\end{definition}

A mechanism $\M$ is $\eps$-EEIC iff no agent can gain more than $\eps$ ex post regret, in expectation over all type profiles $t\in \T$ (where ex post regret is the amount by which an agent's utility can be improved by misreporting to some $\hat{t}_i$ given knowledge of $t$,
instead of reporting its true type $t_i$).
A 0-EEIC mechanism is essentially DSIC.\footnote{For discrete type settings, 0-EEIC is exactly DSIC. For the continuous type case, a 0-EEIC mechanism is DSIC up to zero measure events.}

We can also consider an interim version of $\eps$-EEIC, termed as $\eps$-expected interim IC ($\eps$-EIIC),  which is  defined as
\begin{eqnarray*}
\E_{t_i\sim \F_i}\left[\E_{t_{-i}\sim \F_{-i}} \left[u_i(t_i, \M(t_i; t_{-i}))\right]\right] \geq \E_{t_i\sim \F_i}\left[\max_{t'_i\in\T_i} \E_{t_{-i}\sim \F_{-i}} \left[u_i(t_i, \M(t'_i; t_{-i}))\right]\right] - \eps
\end{eqnarray*}
\noindent All our results for $\eps$-EEIC to BIC transformation hold for $\eps$-EIIC mechanism. Indeed, we prove that any $\eps$-EEIC mechanism is also $\eps$-EIIC in Lemma~\ref{lem:EEIC-decomposition} in the  Appendix.

Another important property of mechanism design is {\em individual rationality} (IR), and we define two standard versions of IR (ex-post/interim IR) in Appendix~\ref{app:omitted-definitions}.
The transformation that we provide from $\eps$-BIC/$\eps$-EEIC to BIC preserves individual rationality: if the original mechanism is interim IR then the mechanism achieved after transformation is interim IR, and if the the original mechanism is ex-post IR then the mechanism achieved after transformation is ex-post IR.

For a mechanism $\M$, let $R^\M(\F)$ and $W^\M(\F)$ represent the expected revenue and social welfare, respectively, when agent types are sampled from $\F$ and they play $\M$ \emph{truthfully} \footnote{In this paper, we consider the revenue and welfare performance of the untruthful mechanisms with truthful reports, which is commonly used in the literature. It is an interesting future direction to consider the performance of untruthful mechanisms under equilibrium reporting.}.  This definition applies equally to an IC or non-IC mechanism.

\begin{definition}[Expected Social Welfare and Revenue]\label{def:rev-social-welfare}
For a mechanism $\M = (x, p)$ with agents' types drawn from distribution $\F$, the expected revenue for truthful reports  is $R^\M(\F) = \E_{t\sim \F}[\sum_{i=1}^n p_i(t)]$, and the expected social welfare for truthful reports is $W^\M(\F) = \E_{t\sim \F}[\sum_{i=1}^n v_i(t_i, x(t))]$. 
\end{definition}

We focus on welfare-preserving transforms that provide negligible revenue loss.
\begin{definition}[Welfare-preserving Transformation with Negligible Revenue Loss] Given an $\eps$-BIC/\newnew{$\eps$-EEIC} mechanism $\M$ over type distribution $\F$, a welfare-preserving transform that provides negligible revenue loss outputs a mechanism $\M'$ such that, $W^{\M'}(\F) \geq W^{\M}(\F)$ and $R^{\M'}(\F) \geq R^{\M}(\F) - r(\eps)$, where $r(\eps) \rightarrow 0$ as $\eps\rightarrow 0$.
\end{definition}

\subsection{Previous $\eps$-BIC to BIC transformations}

There are existing algorithms for transforming any $\eps$-BIC mechanism to an exactly BIC mechanism with only negligible revenue loss~\citep{DasWeinberg12, Rubinstein18, Cai19}. The central tools and reductions in these papers build upon the method of  \emph{replica-surrogate matching}~\citep{Hartline10, Hartline11, BeiHuang11}.
Here we briefly introduce replica-surrogate matching and its application to an $\eps$-BIC to BIC transformation.
\smallskip

\noindent{\bf Replica-surrogate matching. } For each agent $i$, construct a bipartite graph $G_i = (\mathcal{R}_i\cup \mathcal{S}_i, E)$. The nodes in $\mathcal{R}_i$ are called {\em replicas}, which are types sampled i.i.d.~from the type distribution of agent $i$, $\F_i$. The nodes in $\mathcal{S}_i$ are called {\em surrogates}, and also sampled from $\F_i$. In particular, the true type $t_i$ is added in $\mathcal{R}_i$. %
There is an edge between each replica and each surrogate. The weight of the edge between a replica $r^{(j)}_i$ and a surrogate $s^{(k)}_i$ is induced by the mechanism, and defined as 
\begin{eqnarray}\label{eq:replica-surrogate-weight}
w_i(r_i^{(j)}, s^{(k)}) %
= E_{t_{-i}\in \F_{-i}}\left[v_i(r_i^{(j)}, x(s_i^{(k)}, t_{-i}))\right] - (1-\eta)\cdot \E_{t_{-i}\in \F_{-i}}\left[p_i(s_i^{(k)}, t_{-i})\right].
\end{eqnarray}

The replica-surrogate matching computes the maximum weight matching in $G_i$. %
\smallskip

\noindent{\bf $\eps$-BIC to BIC transformation by Replica-Surrogate Matching~\citep{DasWeinberg12}.}
We briefly describe this transformation,  deferring the details to  Appendix~\ref{app:replica-surrogate-mechanism}. Given a mechanism $\M=(x, p)$, this transformation constructs a bipartite graph between replicas (include the true type $t_i$) and surrogates, as described above.
The approach then runs VCG matching to compute the maximum weighted matching for this bipartite graph, and charges each agent its VCG payment.
For unmatched replicas in the VCG matching, the method randomly matches a surrogate. %
Let $\M'= (x, (1-\eta)p)$ be the modified mechanism.
If the true type $t_i$ is matched to a surrogate $s_i$, then  agent $i$ uses $s_i$ to compete in $\M'$. The outcome of $\M'$ is $x(s)$, given matched surrogate profile $s$, and the payment of agent $i$ (matched in VCG matching) is $(1-\eta)p_i(s)$ plus the VCG payment from the VCG matching, where $\eta$ is the parameter in replica-surrogate matching . If $t_i$ is not matched in the VCG matching, the agent gets nothing and pays zero.

This replica-surrogate matching transform does not preserve welfare.
\newnew{Indeed, the replica-surrogate matching transformation must suffer welfare loss in some cases.\footnote{The previous $\eps$-BIC to BIC transformations~\citep{DasWeinberg12, Rubinstein18, Cai19} don't state the welfare loss guarantee clearly. Consider Example~\ref{example:better-result} shown in Section~\ref{sec:contribution}, the original $\eps$-BIC mechanism already maximizes welfare and the optimal allocation is unique, any unmatched type in replica-surrogate matching creates a welfare loss. Particularly, the welfare loss is \emph{unbounded} when (inappropriately) choosing $\eta < \frac{\eps}{\sqrt{m} - 1}$ in replica-surrogate matching.}
}
Turning to revenue, the revenue loss of the replica-surrogate matching mechanism relative to the orginal mechanism $\M$ is guaranteed to be at most %
$\eta\Rev(\M) + \O\left(\frac{n\eps}{\eta}\right)$~\citep{DasWeinberg12, Rubinstein18}, and has both a multiplicative and an additive term.
\citet{Cai19} proposes a polynomial time algorithm for performing this transform
with only sample access to the type distribution and query access to the original $\eps$-BIC mechanism. \newnew{The transform extends replica-surrogate matching and \emph{Bernoulli factory} techniques proposed by~\cite{Dughmi17} to handle negative weights in the bipartite graph
and provides the same revenue property as the previous work~\citep{DasWeinberg12, Rubinstein18}, without preserving social welfare.\footnote{\citet{Dughmi17} propose a general transformation from any black-box algorithm $\mathcal{A}$ to a BIC mechanism that only incurs negligible loss of welfare, with only polynomial number queries to $\mathcal{A}$, by using Bernoulli factory techniques. This approach has no guarantee on the revenue loss. \citet{Cai19} generalize Bernoulli factory techniques in the replica-surrogate matching to transform any $\eps$-BIC mechanism to a BIC mechanism that only incurs negligible loss of revenue, with polynomial number queries to the original $\eps$-BIC mechanism and polynomial number samples from the type distribution.
}} %
In this work, we assume oracle access to the interim quantities of the original $\eps$-BIC mechanism, following 
the model of \citet{Hartline10, Hartline11, BeiHuang11, DasWeinberg12, Rubinstein18}. How to generalize the proposed transform to the setting that only has sample access to the type distribution and runs in polynomial time will be an interesting future work.

The black-box reduction of~\citet{BeiHuang11} focuses on preserving welfare only. Indeed, it can be regarded as a special case of this replica-surrogate matching method, where the weight of the bipartite graph only depends on the valuations and not the prices ($\eta = 1$ in Eq.~(\ref{eq:replica-surrogate-weight})), 
and the replicas and surrogates are both $\T_i$ (there is no sampling for replicas and surrogates).
For this reason, the transform described in~\citet{BeiHuang11} can preserve social welfare  but may provide arbitrarily bad revenue (see Example~\ref{example:better-result}). %

\subsection{Our Contributions}\label{sec:contribution}

We first state the main result of the paper, which provides a welfare-preserving transform from approximate BIC to exact BIC with negligible revenue loss. This result holds for the general mechanism design setting with $n\geq 1$ agents and independent private types and is not restricted to allocation problems.
\begin{main_theorem}[Theorem~\ref{thm:multi-agent-BIC}]\label{thm:main-thm-1}
With $n\geq 1$ agents and independent private types, and an $\eps$-BIC and IR mechanism $\M$ that achieves $W$ expected social welfare and $R$ expected revenue given truthful reports, there exists a BIC and IR mechanism $\M'$ that achieves at least $W$ social welfare and $R - \sum_{i=1}^n|\T_i|\eps$ revenue. The transformation is (ex ante) allocation-invariant. Given an oracle access to the interim quantities of $\M$, the running time of the transformation from $\M$ to $\M'$ is at most $\mathtt{poly}(\sum_i|\T_i|
)$. %
\end{main_theorem}

The transformation works directly on the type graph of each agent, and it is this that allows us to maintain social welfare--- indeed, we may even improve social welfare in our transformation. In contrast, the transformation from~\citet{BeiHuang11} can incur unbounded revenue loss (see Example~\ref{example:better-result}, in which it loses all revenue)
and existing approaches~\citep{DasWeinberg12, CaiZhao17, Rubinstein18, Cai19} with negligible revenue loss can lose social welfare (see Example~\ref{example:better-result}).

Choosing $\eta = \sqrt{\eps}$, the revenue loss of existing transforms~\citep{DasWeinberg12, CaiZhao17, Rubinstein18, Cai19} is at most $\sqrt{\eps}\Rev(\M) + O(n\sqrt{\eps})$, 
 with  both a multiplicative and an additive-loss in revenue,
while our revenue loss is additive.
In  the case that the original revenue, $\Rev(\M)$, is order-wise smaller than the number of types, i.e., $\Rev(\M) = o(\sum_i |\T_i|)$, the existing transforms provide a better revenue bound (at some cost of welfare loss). But when the
revenue is relatively larger than the number of types, i.e., $\Rev(\M) = \Omega(\sum_i|\T_i|)$,
our transformation can achieve strictly better revenue than these earlier approaches
while also preserving welfare. %

Before describing our techniques we illustrate these properties through a single agent, two outcome example in Example~\ref{example:better-result}. We show that even for the case that $\Rev(M) = o(\sum_{i} |\T_i|)$, our transformation can
strictly outperform existing transforms w.r.t revenue loss.
\begin{example}\label{example:better-result}
Consider a single agent with $m$ types, $\T=\{t^{(1)}, \cdots, t^{(m)}\}$, where the type distribution is uniform. Suppose there are two outcomes, the agent with type $t^{(j)} (j=1,\ldots, m-1)$ values outcome 1 at 1 and values outcome 2 at 0. The agent with type $t^{(m)}$ values outcome 1 at $1+\eps$  and outcome 2 at $\sqrt{m}$.
The mechanism $\M$ we consider is: if the agent reports type $t^{(j)}, j\in[m-1]$, $\M$ gives outcome 1 to the agent with a price of 1, and if the agent reports type $t^{(m)}$, $\M$ gives outcome 2 to the agent with a price of $\sqrt{m}$. $\M$ is $\eps$-BIC, because the agent with type $t^{(m)}$ has a regret $\eps$. The expected revenue achieved by $\M$ is $1 + \frac{\sqrt{m}-1}{m}$. In addition, $\M$ maximizes social welfare, $1+\frac{\sqrt{m} - 1}{m}$.

Our transformation decreases the payment of type $t^{(m)}$ by $\eps$ for a loss of $\frac{\eps}{m}$ revenue and preserves the social welfare. 

The transformation by~\citet{BeiHuang11} preserves the social welfare, however, the VCG payment (envy-free prices) is $0$ for each type. Therefore,~\citet{BeiHuang11}'s approach loses all revenue.

Moreover, the approaches that make use of replica-surrogate matching~\citep[e.g.]{DasWeinberg12, Rubinstein18, Cai19} lose at least $\frac{\eps}{m} + \frac{\eps}{\sqrt{m} - 1}$ revenue, which is about $(\sqrt{m}+1)$ times larger than the revenue loss of our transformation. %
We argue this claim by a case analysis,
\begin{itemize}
\item If $\eta \geq \frac{\eps}{\sqrt{m}-1}$, the VCG matching is the identical matching and the VCG payment is $0$ for each type. In total, the agent loses at least $\eta\cdot \frac{\sqrt{m} + m -1}{m} \geq \frac{\eps}{m} + \frac{\eps}{\sqrt{m} - 1}$ expected revenue.

\item If $\eta < \frac{\eps}{\sqrt{m}-1}$, the agent with type $t^{(m)}$ will be assigned outcome 1 ($t^{(m)}$ is matched to some $t^{(j)}, j\in[m-1]$, in VCG matching) and the VCG payment is $\eta$. Thus, type $t^{(m)}$ loses at least $\sqrt{m} - (1-\eta) -\eta = \sqrt{m}-1$ revenue. For any type $t^{(j)}, j\in[m-1]$, if $t^{(j)}$ is matched in VCG matching, the VCG payment is $0$, since it will be matched to another type $t^{(k)}, k\in[m-1]$. Each type $t^{(j)}, j\in[m-1]$ loses at least $\eta$ revenue. Overall the agent loses at least $\frac{\sqrt{m}-1}{m}$ expected revenue. In addition, since the type $t^{(m)}$ is assigned outcome 1, we lose at least $\frac{\sqrt{m}-1-\eps}{m}$ expected social welfare.
\end{itemize}

Moreover, there is a chance that a type is not matched, in which case the social welfare is reduced.
\end{example}

Our transformation satisfies also satisfies an appealing allocation-invariance property (see Definition~\ref{def:allocation-invariant}). Given an $\eps$-BIC mechanism $\M = (x, p)$, the transform outputs a BIC mechanism $\M'= (x', p')$ that satisfies $\sum_{t\in \T} f(t)x'(t) = \sum_{t\in \T} f(t)x(t)$. As noted above, this property would be of interest, for example, to  a principal who is operating the logistics for provisioning goods sold through the mechanism. Because of allocation-invariance, the principal knows that the distribution on goods sold is unchanged as a result of the transform and thus logistical aspects in regard to inventory storage are unchanged. %
The previous transformations~\citep{BeiHuang11, DasWeinberg12, Rubinstein18, Cai19} don't satisfy this allocation-invariance property.

We also support $\eps$-expected ex-post IC ($\eps$-EEIC), which is motivated by work on the use of  machine learning to achieve approximate IC mechanisms in multi-dimensional settings~\citep{DuettingFJLLP12,FengEtAl18, Duetting19}.
In comparison with $\eps$-BIC, the $\eps$-EEIC metric only guarantees at most $\eps$ ex-post gain in expectation over type profiles, with no interim guarantee for any particular type. It is incomparable in strength with $\epsilon$-BIC because $\eps$-EEIC also strengthens $\eps$-BIC in working with ex-post regret rather than interim regret. %
Our second main result shows how to transform an $\eps$-EEIC mechanism to a BIC mechanism. \dcpadd{For this, we need the additional assumption of a uniform type distribution and prove that this is necssary to achieve a transform with suitable properties.}
\begin{main_theorem}[Informal Theorem~\ref{thm:impossibility-EEIC} and Theorem~\ref{thm:multi-agent-BIC}]
For $n \geq 1$ agents with independent uniform type distribution, our $\eps$-BIC to BIC transformation can be applied to an $\eps$-EEIC mechanism and all results in \dcpadd{Main Result 1} hold here.
For a non-uniform type distribution, we show an impossibility result for an $\eps$-EEIC to BIC, welfare-preserving transformation with only negligible revenue loss, even for the single agent case.
\end{main_theorem}

Moreover, we also argue that our revenue loss bounds are tight given the requirement to maintain social welfare. This holds for both $\eps$-BIC mechanisms and $\eps$-EEIC mechanisms.
\begin{main_theorem}[Informal Theorem~\ref{thm:lower-bound-revenue-loss} and Theorem~\ref{thm:lb-revenue-loss-multiple-agents}]
There exists an $\eps$-BIC/$\eps$-EEIC and IR mechanism for $n\geq 1$ agents with independent uniform type distribution, for which any welfare-preserving transformation must suffer $\Omega(\sum_i |\T_i| \eps)$ revenue loss.
\end{main_theorem}

We also apply the transform to automated mechanism design in Section~\ref{sec:application},  considering both a linear-programming and machine learning framework and looking to maximize a linear combination of expected revenue and social welfare, i.e., $\mu_\lambda(\M, \F) = (1-\lambda)R^\M(\F) + \lambda W^{\M}(\F)$,
for some $\lambda\in[0, 1]$ and type distribution $\F$. %
We summarize the result of this application.
\begin{main_theorem}[Informal Theorem~\ref{thm:application-LP} and Theorem~\ref{thm:application-regretnet}]
For $n$ agents with independent type distribution $\times_{i=1}^n \F_i$ on $\T=\T_1 \times\cdots\times \T_n$ and an $\alpha$-approximation LP algorithm $\mathtt{ALG}$ to output an $\eps$-BIC ($\eps$-EEIC) and IR mechanism $\M$ on $\F$ with $\mu_\lambda(\M, \F) \geq \alpha \mathtt{OPT}$, there exists a BIC and IR mechanism $\M'$, s.t., $\mu_\lambda(\M', \F) \geq \alpha \mathtt{OPT} - (1-\lambda)\sum_{i=1}^n \vert \T_i\vert \eps$. Given oracle access to the interim quantities of $\M$, the running time to output the mechanism $\M'$ is at most $\mathtt{poly}(\sum_{i=1} |\T_i|, rt_{\mathtt{ALG}}(x))$, where $rt_{\mathtt{ALG}}(\cdot)$ is the running time of $\mathtt{ALG}$ and $x$ is the bit complexity of the input. Similar results hold for a machine-learning based approach, in a PAC learning manner.
\end{main_theorem}

Compared with the previous transformations that are able to achieve negligible revenue loss~\citep{DasWeinberg12,Rubinstein18, Cai19}, our transformation achieves a better blended objective of welfare and revenue when $\lambda$ is close to $1$ since we preserve welfare of the original mechanism after transformation.

\subsection{Our Techniques}\label{sec:our-technique}

Instead of constructing a bipartite replica-surrogate graph, our transformation makes use of a directed, weighted type graph, one for  each agent. %
For simplicity of exposition, we can consider a single agent with a uniform type distribution.

Given an $\eps$-BIC mechanism, $\M$, we construct a graph $G=(\T, E)$, where each node represents a possible type of the agent and there is an edge from node $t^{(j)}$ to $t^{(k)}$ if the output of the mechanism for type $t^{(k)}$ %
is weakly preferred by the agent for true type $t^{(j)}$ in $\M$, i.e. $u(t^{(j)}, \M(t^{(k)})) \geq u(t^{(j)}, \M(t^{(j)}))$.
The weight $w_{jk}$ of edge $(t^{(j)}, t^{(k)})$ is defined as the {\em regret} of  type $t^{(j)}$ by not misreporting $t^{(k)}$, i.e.,
\begin{eqnarray}\label{eq:weight-type-graph}
w_{jk} = u(t^{(j)}, \M(t^{(k)})) - u(t^{(j)}, \M^\eps(t^{(j)})).
\end{eqnarray}

The transformation method then iterates over the following two steps, constructing a transformed mechanism from the original mechanism. We briefly introduce the two steps here and defer to Figure~\ref{transform:single-agent-uniform} for detailed description.
\smallskip

{\bf Step 1.} If there is a cycle $\C$ in the type graph with at least one positive-weight edge, then all types in this cycle weakly prefer their descendant in the cycle and one or more strictly prefers their descendant.  In this case, we ``\emph{rotate}" the outcome and payment of types against the direction of the cycle, to let each type receive a weakly better outcome compared with its current outcome. We repeat Step 1 until all cycles in the type graph are removed.

{\bf Step 2.}
We pick a source node, if any, with a positive-weight outgoing edge (and thus regret for truthful reporting). We decrease the payment made by this source node, as well as decreasing the payment made by each one of its ancestors (note the lack of cycles at this point) by the same amount, until we create a new edge in the type graph with weight zero, such that the modification to payments is about to increase regret for some type. If at any point we create a cycle, we move to Step 1. Otherwise, we repeat Step 2 until there are no source nodes with positive-weight, outgoing edges. %
\medskip

The algorithm works on the type graph induced by the original, approximately IC mechanism, $\M$, and  directly modifies the mechanism for each type, to make the mechanism IC. This allows the transformation to preserve welfare and provides negligible revenue loss. Step 2 has no effect on welfare, since it only changes (interim) payment for each type. Step 1 is designed to remove cycles created in Step 2 so that we can run Step 2, while preserving welfare simultaneously. Both steps reduce the total weight of the type graph, which is equivalent to reducing the regret in the mechanism to make it IC.
We illustrate the transform in Fig.~\ref{fig:single-agent-uniform}. %
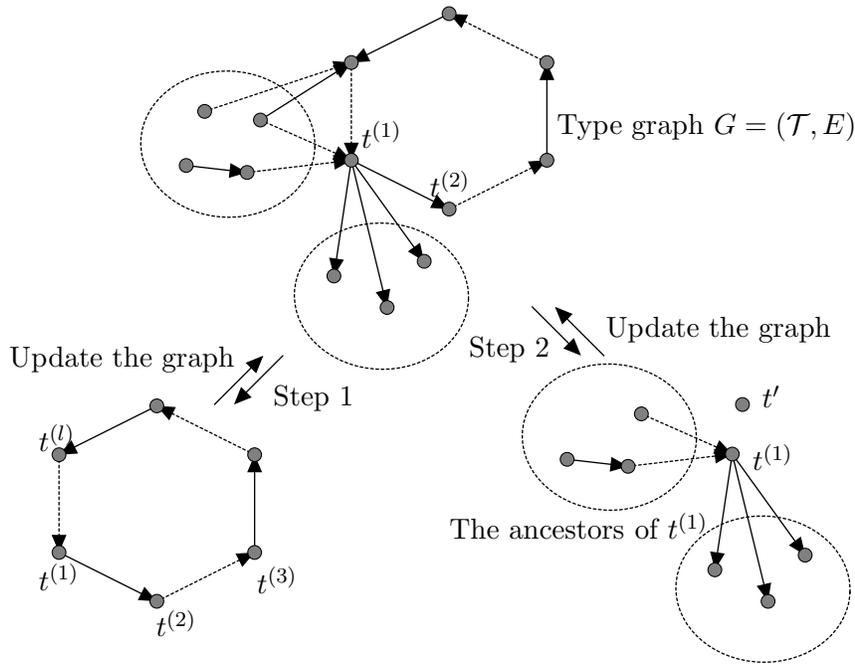
\begin{figure*}[h]
\centering
\begin{tikzpicture}[scale=1.3, line cap=round,line join=round,>=triangle 45,x=1.0cm,y=1.0cm]
\clip(-6.,6.5) rectangle (3.,14.);
\draw (-4.6,7.49913859274312) node[anchor=north west] {$t^{(2)}$};
\draw (-3.56709372991065,7.943337836726545) node[anchor=north west] {$t^{(3)}$};
\draw (-5.8,7.975643236288975) node[anchor=north west] {$t^{(1)}$};
\draw (-5.8,9.4) node[anchor=north west] {$t^{(l)}$};
\draw [->,line width=0.5pt] (-2.5,12.) -- (-1.5,11.5);
\draw [->,line width=0.5pt] (-0.5,12.) -- (-0.5,13.);
\draw [->,line width=0.5pt] (-1.5,13.5) -- (-2.5,13.);
\draw [->,line width=0.5pt,dash pattern=on 1pt off 1pt] (-1.5,11.5) -- (-0.5,12.);
\draw [->,line width=0.5pt,dash pattern=on 1pt off 1pt] (-0.5,13.) -- (-1.5,13.5);
\draw [->,line width=0.5pt,dash pattern=on 1pt off 1pt] (-2.5,13.) -- (-2.5,12.);
\draw [rotate around={-2.1210963966614047:(-3.7627370618230502,12.164608276836091)},line width=0.5pt,dash pattern=on 1pt off 1pt] (-3.7627370618230502,12.164608276836091) ellipse (0.8905164963388845cm and 0.7477949646867164cm);
\draw [->,line width=0.5pt] (-3.429851853529851,12.411587624924588) -- (-2.5,13.);
\draw [->,line width=0.5pt,dash pattern=on 1pt off 1pt] (-3.565869465520619,11.874675998645243) -- (-2.5,12.);
\draw [->,line width=0.5pt,dash pattern=on 1pt off 1pt] (-3.429851853529851,12.411587624924588) -- (-2.5,12.);
\draw [->,line width=0.5pt,dash pattern=on 1pt off 1pt] (-4.,12.5) -- (-2.5,13.);
\draw [->,line width=0.5pt] (-4.188686952004661,11.946264215482488) -- (-3.565869465520619,11.874675998645243);
\draw [rotate around={-2.121096396661412:(-2.1919251120115626,10.608864313856833)},line width=0.5pt,dash pattern=on 1pt off 1pt] (-2.1919251120115626,10.608864313856833) ellipse (0.8905164963388846cm and 0.7477949646867171cm);
\draw [->,line width=0.5pt] (-2.5,12.) -- (-2.134105128775695,10.49302341368639);
\draw [->,line width=0.5pt] (-2.5,12.) -- (-2.678175576738766,10.815170389453996);
\draw [->,line width=0.5pt] (-2.5,12.) -- (-1.7546875795382901,10.965505644812213);
\draw (-2.5,12.5) node[anchor=north west] {$t^{(1)}$};
\draw (-1.8,12) node[anchor=north west] {$t^{(2)}$};
\draw [->,line width=0.5pt] (-5.488861080180847,7.986085603434901) -- (-4.488861080180845,7.486085603434901);
\draw [->,line width=0.5pt,dash pattern=on 1pt off 1pt] (-4.488861080180845,7.486085603434901) -- (-3.488861080180844,7.986085603434901);
\draw [->,line width=0.5pt] (-3.488861080180844,7.986085603434901) -- (-3.488861080180844,8.986085603434901);
\draw [->,line width=0.5pt,dash pattern=on 1pt off 1pt] (-3.488861080180844,8.986085603434901) -- (-4.488861080180845,9.486085603434901);
\draw [->,line width=0.5pt] (-4.488861080180845,9.486085603434901) -- (-5.488861080180847,8.986085603434901);
\draw [->,line width=0.5pt,dash pattern=on 1pt off 1pt] (-5.488861080180847,8.986085603434901) -- (-5.488861080180847,7.986085603434901);
\draw [->,line width=0.5pt] (-3.9,9.5) -- (-3.4,10.0);
\draw [->,line width=0.5pt] (-3.2,10.0) -- (-3.7,9.5);
\draw (-3.4,9.8) node[anchor=north west] {Step 1};
\draw [rotate around={-2.1210963966614105:(0.13865557869089434,9.166587568939116)},line width=0.5pt,dash pattern=on 1pt off 1pt] (0.13865557869089434,9.166587568939116) ellipse (0.8905164963388553cm and 0.7477949646866925cm);
\draw [->,line width=0.5pt,dash pattern=on 1pt off 1pt] (0.3299045320249989,8.86746623143583) -- (1.395773997545618,8.992790232790588);
\draw [->,line width=0.5pt,dash pattern=on 1pt off 1pt] (0.46592214401576715,9.404377857715176) -- (1.395773997545618,8.992790232790588);
\draw [->,line width=0.5pt] (-0.29291295445904253,8.939054448273076) -- (0.3299045320249989,8.86746623143583);
\draw [rotate around={-2.121096396661411:(1.7094675285023833,7.6108436059598565)},line width=0.5pt,dash pattern=on 1pt off 1pt] (1.7094675285023833,7.6108436059598565) ellipse (0.8905164963388637cm and 0.7477949646866994cm);
\draw [->,line width=0.5pt] (1.395773997545618,8.992790232790588) -- (1.761668868769923,7.485813646476977);
\draw [->,line width=0.5pt] (1.395773997545618,8.992790232790588) -- (1.2175984208068522,7.807960622244584);
\draw [->,line width=0.5pt] (1.395773997545618,8.992790232790588) -- (2.141086418007328,7.958295877602801);
\draw [->,line width=0.5pt] (-0.6507275128210431,10.5) -- (-0.15072751282104313,10.);
\draw [->,line width=0.5pt] (0.08290013205157373,10.) -- (-0.41709986794842624,10.5);
\draw (-1.4,10.3) node[anchor=north west] {Step 2};
\draw (-6.1,10.2) node[anchor=north west] {Update the graph};
\draw (0,10.5) node[anchor=north west] {Update the graph};
\draw (-0.5,12.6) node[anchor=north west] {Type graph $G=(\T, E)$};
\draw (1.5,9.2) node[anchor=north west] {$t^{(1)}$};
\draw (-1.6,8.5) node[anchor=north west] {The ancestors of $t^{(1)}$};
\draw (1.6,9.8) node[anchor=north west] {$t'$};
\begin{scriptsize}
\draw [fill=gray] (-2.5,13.) circle (2.0pt);
\draw [fill=gray] (-1.5,11.5) circle (2.0pt);
\draw [fill=gray] (-0.5,12.) circle (2.0pt);
\draw [fill=gray] (-0.5,13.) circle (2.0pt);
\draw [fill=gray] (-1.5,13.5) circle (2.0pt);
\draw [fill=gray] (-2.5,12.) circle (2.0pt);
\draw [fill=gray] (-4.,12.5) circle (2.0pt);
\draw [fill=gray] (-3.429851853529851,12.411587624924588) circle (2.0pt);
\draw [fill=gray] (-4.188686952004661,11.946264215482488) circle (2.0pt);
\draw [fill=gray] (-3.565869465520619,11.874675998645243) circle (2.0pt);
\draw [fill=gray] (-2.678175576738766,10.815170389453996) circle (2.0pt);
\draw [fill=gray] (-1.7546875795382901,10.965505644812213) circle (2.0pt);
\draw [fill=gray] (1.5,9.5) circle (2.0pt);
\draw [fill=gray] (-2.134105128775695,10.49302341368639) circle (2.0pt);
\draw [fill=gray] (-5.488861080180847,7.986085603434901) circle (2.0pt);
\draw [fill=gray] (-3.488861080180844,7.986085603434901) circle (2.0pt);
\draw [fill=gray] (-3.488861080180844,8.986085603434901) circle (2.0pt);
\draw [fill=gray] (-5.488861080180847,8.986085603434901) circle (2.0pt);
\draw [fill=gray] (-4.488861080180845,7.486085603434901) circle (2.0pt);
\draw [fill=gray] (-4.488861080180845,9.486085603434901) circle (2.0pt);
\draw [fill=gray] (1.395773997545618,8.992790232790588) circle (2.0pt);
\draw [fill=gray] (0.46592214401576715,9.404377857715176) circle (2.0pt);
\draw [fill=gray] (-0.29291295445904253,8.939054448273076) circle (2.0pt);
\draw [fill=gray] (0.3299045320249989,8.86746623143583) circle (2.0pt);
\draw [fill=gray] (1.2175984208068522,7.807960622244584) circle (2.0pt);
\draw [fill=gray] (2.141086418007328,7.958295877602801) circle (2.0pt);
\draw [fill=gray] (1.761668868769923,7.485813646476977) circle (2.0pt);
\end{scriptsize}
\end{tikzpicture}

\caption{Visualization of the transformation for a single agent with a uniform type distribution: we start from a type graph $G(\T, E)$, where each edge $(t^{(1)}, t^{(2)})$ represents the agent weakly prefers the allocation and payment of type $t^{(2)}$ rather than his true type $t^{(1)}$. The weight of each edge is denoted in Eq.~(\ref{eq:weight-type-graph}). In the graph, we use solid lines to represent the positive-weight edges, and dashed lines to represent  zero-weight edges.  We first find a shortest cycle, and rotate the allocation and payment along the cycle  and update the graph (Step 1). We keep doing Step 1 to remove all  cycles. Then we pick  a source node $t^{(1)}$, and decrease the payment of type $t^{(1)}$ and all the ancestors of $t^{(1)}$ until we reduce the weight of one outgoing edge from $t^{(1)}$ to zero or we create a new zero-weight edge from $t'$ to $t^{(1)}$ or one of the ancestors of $t^{(1)}$ (Step 2).
\label{fig:single-agent-uniform}}
\vspace{-10pt}
\end{figure*}

For a single agent with non-uniform type distribution, we handle the unbalanced density probability of each type by redefining the type graph, where the weight of the edge in type graph is weighted by the product of the probability of the two nodes that are incident to an edge.  We propose a new Step 1 by introducing \emph{fractional rotation}, such that for each cycle in the type graph, we rotate the allocation and payment with a fraction for any type $t^{(j)}$ in the cycle. By carefully choosing the fraction for each type in the cycle, we can argue that our transformation preserves welfare and provides negligible revenue loss. %

The multi-agent setting reduces to the single-agent case, building a type graph for each agent induced by the interim rules (see Appendix~\ref{app:multi-agent-BIC} for the construction of this type graph).
With oracle access to the interim quantities of the original mechanism, we build the type graph of each agent $i$ in $\mathtt{poly}(|\T_i|)$ time.
We then apply the transform for each type graph of agent $i$, induced by the interim rules.

This is analogous to a replica-surrogate matching approach, which also defines the weights between replicas and surrogates by interim rules and runs the replica-surrogate matching for the reported type of each agent. Replica-surrogate matching uses this sampling technique to make the distribution of reported types of each agent equal to the distribution of the true type. In comparison, Steps 1 and 2 of our transform leave the type distribution unchanged, so that the transform attains this property for free. Then we can apply our transformation for each type graph separately. The new challenge in our transformation is feasibility, i.e., establishing consistency of the agent-wise rotations to interim quantities. We show the transformation for each type graph guarantees feasibility by appeal to Border's lemma~\citep{Border91}. Our transformation can also be directly applied to an $\eps$-EEIC mechanism in the case that each agent has an independent uniform type distribution.\footnote{This need to transform an infeasible, IC mechanism into a feasible and IC mechanism also arises in \citet{Narasimhan_UAI16}, who use a method from~\citet{Hashimoto18} to correct for feasibility violations that result from statistical machine learning while preserving strategy-proofness.} %

\section{Warm-up: Single agent with Uniform Type Distribution}\label{sec:single-agent-uniform}

In this section, we consider the  case of a single agent and a uniformly distributed type distribution $\F$, i.e. $\forall j \in[m], f(t^{(j)}) = \frac{1}{m}$. %
Even for this simple case, the proof is non-trivial. Moreover, the technique for this simple case can be extended to handle more intricate cases. The main result in this section is Theorem~\ref{thm:single-agent-uniform},
which makes use of a
constructive proof to modify a $\eps$-EEIC/$\eps$-BIC mechanism to a BIC mechanism.
An interesting observation is that $\eps$-EEIC may only provide  $m\eps$-BIC for  a uniform type distribution, which indicates that transforming $\eps$-EEIC may incur a worse revenue loss bound. However, Theorem~\ref{thm:single-agent-uniform} shows we can achieve the same revenue loss bound for both $\eps$-BIC and $\eps$-EEIC.
\begin{theorem}\label{thm:single-agent-uniform}
Consider a single agent, with $m$ different types $\T=\left\{t^{(1)},t^{(2)},\cdots,t^{(m)}\right\}$, and a uniform type distribution $\F$.
Given an $\epsilon$-EEIC/$\eps$-BIC and IR mechanism $\M$, which achieves $W$ expected social welfare and $R$ expected revenue, there exists a BIC and IR mechanism $\M'$ that achieves at least $W$ expected social welfare and $R- m\eps$ revenue.  Given an oracle access to $\M$, the running time of the transformation from $\M$ to $\M'$ is at most $\mathtt{poly}(|\T|)$.
\end{theorem}
\begin{proof}[Proof Sketch]

\begin{figure*}[h]
\fcolorbox{black}{gray!50!white}{
\parbox{0.98\textwidth}{	
	
{\bf Step 1 (Rotation step).} Given the graph $G$ induced by  $\M=(x, p)$, find the shortest cycle $\C$ in $G$ that contains at least one edge with positive weight. Without loss of generality, we represent $\C = \left\{t^{(1)}, t^{(2)},\cdots, t^{(l)}\right\}$. Then  rotate the allocation and payment rules for these nodes in cycle $\C$. Now we slightly abuse the notation of subscripts, s.t. $t^{(l+1)} = t^{(1)}$. Specifically, the allocation and payment rules for each $t^{(j)}\in \C$,
$x'(t^{(j)}) = x(t^{(j+1)}), p'(t^{(j)}) = p(t^{(j+1)})$.
For other nodes, we keep the allocation and payment rules, i.e. $\forall j \notin [l], x'(t^{(j)}) = x(t^{(j)}), p'(t^{(j)}) = p(t^{(j)})$.  Then we update the mechanism $\M$ by adopting allocation and payment rules $x', p'$ to form a new mechanism $\M'$, and update the graph $G$ (We still use $G$ to represent the updated graph for notation simplicity). If there are no cycles in $G$ that contain at least one positive-weight-edge, move to Step 2. Otherwise, we repeat Step 1.

\noindent{\bf Step 2 (Payment reducing step).} Given the current updated graph $G$ and mechanism $\M'$,  pick up a source node $t$, i.e., a node with no incoming positive-weight edges. Let outgoing edges with positive weights associated with node $t$ be a set of $E_t$, and let $\overline{\eps}_t$ be the minimum non-negative regret of type $t$, i.e.
\begin{eqnarray}\label{eq:eps-bar}
\overline{\eps}_t = \min_{t^{(j)}: (t, t^{(j)})\in E_t} \left[u(t, \M'(t^{(j)})) - u(t, \M'(t))\right]
\end{eqnarray}

Consider the following set of nodes $S_t\subseteq \T$, such that $S_t = \{t\}\cup \{t'\big|t'\in \T \text{ is the ancestor  of } t\}$. The weight zero edge is also counted as a directed edge. 
Denote $\eps_t$ as
\begin{eqnarray}\label{eq:eps}
	\eps_t = \min_{t'\notin S_t, \bar{t}\in S_t}\left[u(t', \M'(t')) - u(t', \M'(\bar{t}))\right]
\end{eqnarray}
Then we decrease the expected payment of all $\bar{t}\in S_t$ by $\min\{\eps_t, \overline{\eps}_t\}$. This process will only create new edges with weight zero. If we create a new cycle with at least one edge with positive weight in $E$, we move to Step 1. Otherwise, we repeat Step 2.
}}
\caption{$\eps$-BIC/$\eps$-EEIC to BIC transformation for single agent with uniform type distribution}
\label{transform:single-agent-uniform}
\end{figure*}

We construct a weighted directed graph $G = (\T, E)$ induced by mechanism $\M$, following the approach shown in Section~\ref{sec:our-technique}.
We apply the iterations of Step 1 and Step 2 (see Fig.~\ref{transform:single-agent-uniform}), to reduce the total weight of edges in $E$ to zero. %

First, we show the  transformation maintains IR, since neither Step 1 nor Step 2 reduces utility. We then argue that the transformation in Fig.~\ref{transform:single-agent-uniform} will reduce the total weight of the graph to zero with no loss of social welfare, and  incur at most $m\eps$ revenue loss. To show this, we prove the following two auxiliary claims in Appendix~\ref{app:single-agent-uniform-claim-1} and \ref{app:single-agent-uniform-claim-2}, respectively.
\smallskip

\begin{claim}\label{claim:single-agent-uniform-claim-1}
Each Step 1 achieves the same revenue and incurs no loss of social welfare, and reduces the total weight of the graph by at least the weights of cycle $\C$.
\end{claim}

\begin{claim}\label{claim:single-agent-uniform-claim-2}
Each Step 2 can only create new edges with zero weight, and does not decrease social welfare. Each Step 2 will reduce the weight of each positive-weight, outgoing edge associated with $t$ by $\min\{\eps_t, \overline{\eps}_t\}$, where $\bar{\eps}_t$ and $\eps_t$ are defined in Eq.~(\ref{eq:eps-bar}) and Eq.~(\ref{eq:eps}) respectively.
\end{claim}

Given the above two claims, we  argue our transformation incurs no loss of social welfare. The transformation only loses revenue at Step 2, for each source node $t$, we decrease at most $m\min\{\eps_t, \overline{\eps}_t\}$ payments over all the types.\footnote{Actually, we can get a slightly tighter bound. Since no cycle exists in the type graph after Step 1, there is at least one node is not the ancestor of $t$. Therefore the revenue decrease is bounded by $(m-1)\min\{\eps_t, \overline{\eps}_t\}$, actually.} In this transformation, after each Step 1 or Step 2, the weight of the outgoing edge of each node $t$ is still bounded by $\max_{j} \left\{u(t, \M(t^{(j)})) - u(t, \M(t))\right\}$. This is because Step 1 does not create new outcome (allocation and payment) and Step 2 will not increase the weight of each edge.
Therefore, in Step 2, we  decrease payments by at most $m\max_{j} \left\{u(t, \M(t^{(j)})) - u(t, \M(t))\right\}$ in order to reduce the weights of all outgoing edges associated with $t$ to zero. Therefore, the total revenue loss in expectation is 
\begin{align*}
&\sum_{t\in \T} \frac{1}{m}\cdot m \max_{j} \left(u(t, \M(t^{(j)})) - u(t, \M(t))\right) \leq m\eps,
\end{align*}
where the inequality is because of the definition of $\eps$-BIC/$\eps$-EEIC mechanism.

\noindent{\bf Running time. } At each Step 1, we strictly reduce the weight of one edge with positive weight to 0 in the graph. The running time of each Step 1 and Step 2 is $\mathtt{poly}(|\T|)$. In total, there are at most $|\T|^2$ edges. Thus, the total running time is $\mathtt{poly}(|\T|)$.
\end{proof}

\subsection{Lower Bound on Revenue Loss}

In the transformation in Figure~\ref{transform:single-agent-uniform}, the revenue loss is bounded by $m\eps$. 
This revenue loss bound is tight  up to a constant factor
while insisting on maintaining social welfare.
\begin{theorem}\label{thm:lower-bound-revenue-loss}
There exists an $\eps$-BIC ($\eps$-EEIC) and IR mechanism $\M$ for a single agent \newnew{with uniform type distribution} for which  any $\eps$-BIC and IR to BIC and IR transformation (without loss of social welfare) must suffer at least $\Omega(m\eps)$ revenue loss.
\end{theorem}
\begin{proof}
Consider a single agent with $m$ types, $\T=\{t^{(1)}, \cdots, t^{(m)}\}$ and $f(t^{(j)})=1/m, \forall j$. There are $m$ possible outcomes. The agent with type $t^{(1)}$ values outcome 1 at $\eps$ and the other outcomes at $0$. For any type $t^{(j)}, j\geq 2$, the agent with type $t^{(j)}$ values outcome $j-1$ at $j\eps$, outcome $j$ at $j\eps$, and the other outcomes at $0$. The original mechanism is: if the agent reports type $t^{(j)}$, gives the outcome $j$ to the agent and charges $j\eps$. There is a $\eps$ regret to an agent with type $t^{(j+1)}$ for not reporting type $t^{(j)}$, thus the mechanism is $\eps$-BIC. Since this $\eps$-BIC mechanism already maximizes social welfare, we cannot change the allocation in the transformation. Thus, we can only change the payment of each type to reduce the regret. Consider the sink node $t^{(1)}$, to reduce the regret of the agent with type $t^{(2)}$ for not reporting $t^{(1)}$, we can increase the payment of type $t^{(1)}$ or decrease the payment of type $t^{(2)}$. However, increasing the payment of type $t^{(1)}$ breaks IR, then we can only decrease the payment of $t^{(2)}$. To reduce the regret between $t^{(2)}$ to $t^{(1)}$, we need to decrease the payment of $t^{(2)}$ at least by $\eps$. After this step, the regret of type $t^{(3)}$ for not reporting $t^{(2)}$ will be at least $2\eps$ and $t^{(2)}$ will be the new sink node. Similarly, $t^{(3)}$ needs to decrease at least $2\eps$ payment (if $t^{(2)}$ increase the payment, it will envy the output of $t^{(1)}$ again). So on and so forth, and in total, the revenue loss is at least $\frac{\eps + 2\eps+\cdots+(m-1)\eps}{m} = \frac{(m-1)\eps}{2}$.
\end{proof}

\subsection{Tighter Bound of Revenue Loss for Settings with  Finite Menus}

In some settings, the total number of possible types of an agent may be very large and yet the menu size can remain relatively small.  In particular, suppose that a mechanism $\M$ has a small number of outputs, i.e., $|\M| = C$ and $C\ll m$, where $m$ is the number of types and $C$ is the menu size.  Given this, we can provide a tighter bound on revenue loss for this setting in the following theorem. The complete proof is deferred to Appendix~\ref{app:vanish-ic-outcome}.
\begin{theorem}\label{thm:vanish-ic-outcome}
Consider a single agent with $m$ different types $\T=\{t^{(1)}, t^{(2)},\cdots, t^{(m)}\}$, sampled from a uniform type distribution $\F$. Given an $\epsilon$-BIC mechanim $\M$ with $C$ different menus ($C \ll m$) that achieves $S$ expected social welfare and $R$ revenue, there exists an BIC mechanism $\M'$ that  achieves at least $S$ social welfare and $R - C\eps$ revenue.
\end{theorem}

\section{Single Agent with General Type Distribution}\label{sec:single-agent-general}

In this section, we consider a setting with a single agent that has a non-uniform type distribution.
A naive idea is that we can ``divide'' a type with a larger probability to several copies of the same type, each with equal probability, and then apply our proof of Theorem~\ref{thm:single-agent-uniform} to get a BIC mechanism. However, this would result in a weak bound on the revenue loss, since we would divide the $m$ types into multiple, small pieces. %
This section is divided into two parts. First we show our transformation for an  $\eps$-BIC mechanism in this setting. Second, we show an impossibility result for an $\eps$-EEIC mechanism, that is, without loss of welfare, no transformation can  achieve negligible revenue loss.

\subsection{$\eps$-BIC to BIC Transformation}
We propose a novel approach for a construction for the case of a single agent with a non-uniform type distribution. The proof is built upon Theorem~\ref{thm:single-agent-uniform}, however, there is a technical difficulty to directly apply the same approach for this non-uniform type distribution case. Since each type has a different probability, we cannot rotate the allocation and payment in the same way as in Step~1 in the proof of Theorem~\ref{thm:single-agent-uniform}. 

We instead redefine the type graph $G=(\T, E)$, where the weight of the edge is now weighted by the product of the probability of the two nodes that are incident to an edge. We also modify the original rotation step shown in Fig.~\ref{transform:single-agent-uniform} in Appendix~\ref{app:single-agent-general}: for each cycle in the type graph, we rotate the allocation and payment with the fraction of $\frac{f(t^{(k)})}{f(t^{(j)})}$ for any type $t^{(j)}$ in the cycle, where $f(t^{(k)})$ is the smallest type probability of the types in the cycle. This step is termed as "fractional rotation step." We summarize the results in Theorem~\ref{thm:single-agent-general} and show the proof in Appendix~\ref{app:single-agent-general}.
\begin{theorem}\label{thm:single-agent-general}
Consider a single agent with $m$ different types, $\T=\left\{t^{(1)},t^{(2)},\cdots,t^{(m)}\right\}$ drawn from a general type distribution $\F$. %
Given an $\eps$-BIC and IR mechanim $\M$ that achieves $W$ expected social welfare and $R$ expected revenue, there exists a BIC and IR mechanism $\M'$ that achieves at least $W$ social welfare and $R - m\eps$ revenue.
\end{theorem}

\noindent{\bf Allocation-invariant Transformation.}
In addition to the welfare and revenue guarantee achieved by this transformation, the transform has another desired property, as defined below.
\begin{definition}[Allocation-invariance property]\label{def:allocation-invariant}	
Two mechanisms $\M = (x, p)$ and $\M' = (x', p')$ are (ex ante) allocation-invariant if and only if $\sum_{t\in \T} f(t)x(t) = \sum_{t\in \T}f(t) x'(t)$.
\end{definition}

For the single agent setting with a general type distribution, the transform only changes the allocation rules in Step 1. Since we use the \emph{fractional rotation} in Step 1, the quantity $\sum_{t\in \T} f(t) x(t)$ is maintained after each Step 1. Then, it is straightforward to show that the transform satisfies this allocation-invariance property.\footnote{By contrast, the previous transformations~\citep{DasWeinberg12, Rubinstein18, Cai19} cannot preserve the distribution of the allocation, even for the single agent and uniform type distribution case.}

\subsection{Impossibility Result for $\eps$-EEIC Transformation}

As mentioned above, given any $\eps$-BIC for a single agent with a general type distribution, we can transform to an exactly BIC mechanis with no loss of welfare and negligible loss of revenue. However, the same claim doesn't hold for $\eps$-EEIC. Theorem~\ref{thm:impossibility-EEIC} shows that  no transformation can achieve negligible revenue loss while insisting on welfare preservation. The proof is provided in Appendix~\ref{app:impossibility-EEIC}.
\begin{theorem}\label{thm:impossibility-EEIC}
There exists a single agent with a non-uniform type distribution,  and an $\eps$-EEIC and IR mechanism, for which there is no IC transformation that preserves social welfare and IR and achieves negligible revenue loss.
\end{theorem}

\section{Multiple Agents with Independent Private Types}\label{sec:multi-agent}
First, we state our positive result for a setting with multiple agents and independent, private types  (Theorem~\ref{thm:multi-agent-BIC}). We assume each agent $i$'s type $t_i$ is independently drawn from  $\F_i$ \newnew{($\F_i$ can be non-uniform)}. Then $\F$ is a product distribution that can be denoted as $\times_{i=1}^n \F_i$. 
The complete proof of the following theorem is shown in Appendix~\ref{app:multi-agent-BIC}.
\begin{theorem}\label{thm:multi-agent-BIC}
With $n$ agents and independent private types, and an $\eps$-BIC and IR mechanism $\M$ that achieves $W$ expected social welfare and $R$ expected revenue, there exists a BIC and IR mechanism $\M'$ that achieves at least $W$ social welfare and $R - \sum_{i=1}^n |\T_i|\eps$ revenue. The same result holds for an $\eps$-EEIC mechanism with multiple agents, in the case that each agent has an independent uniform type distribution. Given an oracle access to the interim quantities of $\M$, the running time of the transformation from $\M$ to $\M'$ is at most $\mathtt{poly}(\sum_i|\T_i|)$.
\end{theorem}

\noindent{\bf Allocation-invariant Transformation.} The transformation
for multiple agents with independent private types is also
allocation-invariant. To prove this, we can observe for $\M=(x, p)$ that
\begin{eqnarray*}
\sum_{t\in \T} f(t)x(t) = \sum_{t_i\in \T_i} f_i(t_i) \cdot\E_{t_{-i}\sim \F_{-i}} [x(t_i, t_{-i})] =  \sum_{t_i\in \T_i} f_i(t_i) X_i(t_i).
\end{eqnarray*}

Then, by Eq.~(\ref{eq:allocation-invariant-multi-agent}) in the proof of Theorem~\ref{thm:multi-agent-BIC} (Appendix~\ref{app:multi-agent-BIC}), we have $\sum_{t\in \T} f(t)x'(t) = \sum_{t\in \T} f(t)x(t)$ for the transformed mechanism $\M'=(x', p')$.

\noindent{\bf Lower bound on revenue loss.} Similarly to single agent
case, we can also prove a lower bound of revenue loss of any
welfare-preserving transformation for multiple agents with independent
private types. We summarize this result in
Theorem~\ref{thm:lb-revenue-loss-multiple-agents}, and show the proof
in Appendix~\ref{app:lb-revenue-loss-multiple-agents}.
\begin{theorem}\label{thm:lb-revenue-loss-multiple-agents}
For any number $n\geq 1$ of agents with independent uniform type distribution,
there exists an $\eps$-BIC/$\eps$-EEIC and IR mechanism, for which any welfare-preserving transformation must suffer at least $\Omega(\sum_i |\T_i| \eps)$ revenue loss.
\end{theorem}

\subsection{Impossibility
  Results}\label{sec:impossibility-multi-agents}

In our main positive result (Theorem~\ref{thm:multi-agent-BIC}), we
assume  independent private types and the target of
transformation is BIC mechanism.  These two assumptions are
near-tight. See Appendix~\ref{app:failure-interdependent-type} and Appendix~\ref{app:failure-dsic-target} for proofs.

\begin{theorem}[Failure of interdependent type]\label{thm:failure-interdependent-type}
There exists an $\eps$-BIC mechanism $\M$ w.r.t an interdependent type distribution $\F$ (see Definition~\ref{def:interdependent-type} in Appendix~\ref{app:omitted-definitions}), such that no BIC mechanism over $\F$ can achieve negligible revenue  loss compared with $\M$.
\end{theorem}

Theorem~\ref{thm:failure-interdependent-type} provides a
counterexample to show that if we allow for interdependent
types,  where the value of one agent depends on the type of another,
there is no way to construct a BIC mechanism without
negligible revenue loss compared with the original $\eps$-BIC
mechanism even if we remove the requirement of welfare preservation.
This leaves an
open question is whether there is a counterexample for an 
$\eps$-BIC transform for correlated, private types.
\begin{theorem}[Failure of DSIC target]\label{thm:failure-dsic-target}
There exists an $\eps$-BIC mechanism $\M$ defined on a type distribution $\F$, such that no DSIC mechanism over $\F$ can achieve negligible revenue loss compared with $\M$.
\end{theorem}

Theorem~\ref{thm:failure-dsic-target} gives an impossibility result
for the setting that we start from an $\eps$-BIC mechanism. %
We leave open the question as to whether it is possible to transform an $\eps$-EEIC mechanism to a DSIC mechanism with zero loss of social welfare and negligible loss of revenue, for multiple agents with independent uniform type distribution.

\vspace{-10pt}
\section{Application to Automated Mechanism Design}\label{sec:application}

In this section, we apply the transform to linear-programming based and machine-learning based approaches to automated mechanism design  (AMD)~\citep{ConitzerS02}, where the mechanism is automatically created for the setting and objective at hand.

\dcpadd{We state the main results for the following,  blended design objective of revenue and welfare}, for a given $\lambda \in [0, 1]$ and type distribution $\F$,
\begin{eqnarray}\label{eq:learning-target-amd}
\mu_\lambda(\M, \F) = (1-\lambda)R^\M(\F) + \lambda W^\M(\F).
\end{eqnarray}

Let $\mathtt{OPT}=\max_{\M: \M \text{ is BIC and IR}} \mu_\lambda(\M, \F)$ be the optimal objective achieved by a BIC and IR mechanism defined on $\F$. We consider two different AMD approaches, an LP-based approach and a machine-learning based
approach.
\smallskip

\noindent\textbf{LP-based AMD.}
\dcpadd{As explained in more detail in Appendix~\ref{app:omitted-application}, an
LP-based approach to BIC mechanism design introduces a decision variable for each outcome and each type profile. }
In practice, the type space of each agent
may be exponential in the number of items for multi-item auctions,
\dcpadd{and the number of type profiles is exponential in the number of agents}. 
To address this challenge,
\dcpadd{it is necessary to} discretize $\T_i$ to a coarser space $\T^+_i, (\vert \T^+_i\vert \ll \vert\T_i\vert)$ and construct the coupled type distribution $\F^+_i$. (e.g., by rounding down to the nearest points in $\T^+_i$, that is, the mass of each point in $T_i$ is associated with the nearest point in $\T_i^+$.)
Then we can apply an LP-based AMD approach for type distribution $\F^+ = (\F^+_1,\cdots,\F^+_n)$.
Even though the LP returns an mechanism defined only on $\T^+$, the mechanism $\M$ can be defined on $\T$, by the same coupling technique. For example, given any type profile $t\in \T$, there is a coupled $t^+\in \T^+$, and the mechanism $\M$ takes $t^+$ as the input. %
This coupling technique makes the mechanism only approximately IC.
Suppose, in particular, that we have an $\alpha$-approximation LP algorithm that  outputs
an $\eps$-BIC and IR mechanism $\M$ over $\F$,
such that $\mu_\lambda(\M, \F)\geq \alpha \mathtt{OPT}$.
By an application of the transform to $\M$, we have the following theorem.
\begin{theorem}[LP-based AMD]\label{thm:application-LP}
  For $n$ agents with independent type distribution $\times_{i=1}^n \F_i$, and an LP-based AMD approach for
  coarsened distribution $\F^+$ on coarsened type space $\T^+$
  that gives an $\eps$-BIC and IR mechanism $\M$ on $\F$,
  with $(1-\lambda) R + \lambda W \geq \alpha \mathtt{OPT}$, for some $\lambda \in[0,1]$, and some $\alpha\in (0,1)$,
then 
  there exists a BIC and IR mechanism $\M'$ such that
\begin{eqnarray*}
\mu_\lambda(\M', \F) \geq \alpha \mathtt{OPT} - (1-\lambda)\sum_{i=1}^n \vert\T_i\vert \eps.
\end{eqnarray*}

Given oracle access to the interim quantities of $\M$ on $\F$ and an $\alpha$-approximation LP solver with running time $rt_{LP}(x)$, where $x$ is the bit complexity of the input, the running time to output the mechanism $\M'$ is at most $\mathtt{poly}(\sum_i |\T_i|, rt_{LP}(\mathtt{poly}(\sum_i |\T^+_i|, \frac{1}{\eps}))$.
\end{theorem}

\smallskip

\noindent\textbf{Machine-learning based AMD.} RegretNet uses an artificial neural network to learn approximately-incentive compatible auctions for multi-dimensional mechanism design~\citep{Duetting19}. See Appendix~\ref{app:omitted-application} for more details of the application of RegretNet to a setting in which the design goal is a blend of revenue and welfare.
RegretNet outputs an $\eps$-EEIC mechanism.
Suppose that  RegretNet is used in a setting with an independent, uniform type distribution $\F$. To train RegretNet, we randomly draw $S$ samples from $\F$ to form a training data $\mathcal{S}$ and train the model on $\mathcal{S}$. %
Let $\mathcal{H}$ be the function space modeled by RegretNet
and suppose a PAC-learner that outputs an $\eps$-EEIC mechanism $\M\in \mathcal{H}$ on $\F$, such that $\mu_\lambda(\M, \F) \geq \sup_{\hat{\M}\in \mathcal{H}}\mu_\lambda(\hat{\M}, \F) - \eps$ holds with probability at least $1-\delta$, by observing $S=S(\eps, \delta)$ i.i.d\ samples from $\F$.
By an application of the transform to $\M$, we have the following theorem.
\begin{theorem}[RegretNet AMD]\label{thm:application-regretnet}
For $n$ agents with independent uniform type distribution $\times_{i=1}^n \F_i$ over $\T=(\T_1,\cdots,\T_n)$, and RegretNet to generate an  $\eps$-EEIC and IR mechanism $\M$ on $\F$ with $\mu_\lambda(\M, \F) \geq \sup_{\hat{\M}\in \mathcal{H}}\mu_\lambda(\hat{\M}, \F) - \eps$ holds with probability at least $1-\delta$, for some $\lambda\in [0,1]$, trained on $S=S(\eps, \delta)$ i.i.d\ samples from $\F$, where $\mathcal{H}$ is the function class modeled by RegretNet, then there exists a BIC and IR mechanism $\M'$, with probability at least $1-\delta$, such that
\begin{eqnarray*}
\mu_\lambda(\M', \F) \geq \sup_{\hat{\M}\in \mathcal{H}}\mu_\lambda(\hat{\M}, \F)  - (1-\lambda)\sum_{i=1}^n \vert \T_i\vert \eps - \eps.
\end{eqnarray*}

Given oracle access to the interim quantities of $\M$ on $\F$ and a PAC-learner with running time $rt_{RegretNet}(x)$, where $x$ is the bit complexity of the input, the running time to output the mechanism $\M'$ is at most $\mathtt{poly}(\sum_i |\T_i|, \eps, rt_{RegretNet}(\mathtt{poly}(S, \frac{1}{\eps}))$.
\end{theorem}

\section{Conclusion}\label{sec:conclusion}

In this paper, we have proposed \dcpadd{the first} $\eps$-BIC to BIC
transformation that achieves negligible revenue loss with no loss in
social welfare.
\newnew{Our transformation differs from the previous replica-surrogate
  matching approaches because we would like to preserve welfare.
  In its place, we directly make use of a directed and weighted type graph (induced by the types' regret), one for each agent. The transformation runs a \emph{fractional rotation step} and a \emph{payment reducing step} iteratively to make the mechanism Bayesian incentive compatible.}
We also proved that the revenue loss bound
of $\sum_{i}|\T_i|\eps$ 
is tight given the requirement that the transform should maintain social welfare. Our transformation also satisfies (ex ante) allocation-invariance property, which cannot be attained by the previous replica-surrogate matching.
In addition, we have demonstrated that the transformation can be
applied to an $\eps$-EEIC mechanism with multiple agents in the case
that each agent has a independent uniform type distribution, and
provided an impossibility result for the case of a non-uniform
distribution and just one agent.

There remain some interesting open questions:
\begin{itemize}
\item Can we design a polynomial time algorithm for an $\eps$-BIC to BIC transformation with negligible revenue loss and without loss of welfare given only query access to the original mechanism and sample access to type distribution?  (Our polynomial time results assume oracle access to the interim quantities.)
\item Is it possible to transform an $\eps$-EEIC mechanism to a DSIC
  mechanism, for multiple agents and with an independent, uniform type distribution, without loss of welfare, and with only negligible revenue loss?
\item If we only focus on the revenue perspective, is it possible to
find an $\eps$-EEIC to DSIC transformation, perhaps even in the
non-uniform case?

\item \newnew{Theorem~\ref{thm:failure-interdependent-type} gives an
    impossibility result for the setting with interdependent type
    distribution. Is it possible to extend our transformation to the
    correlated type distribution setting, or prove an impossibility
    result there?} 
\end{itemize}

\bibliographystyle{ACM-Reference-Format}
\bibliography{approx-ic}

\clearpage
\appendix
\section*{Appendix}

\section{Details of Replica-Surrogate Mechanism}\label{app:replica-surrogate-mechanism}
We show the detailed description of Replica-Surrogate Mechanism in Fig.~\ref{transform:replica-surrogate-matching}.

\begin{figure*}[h]
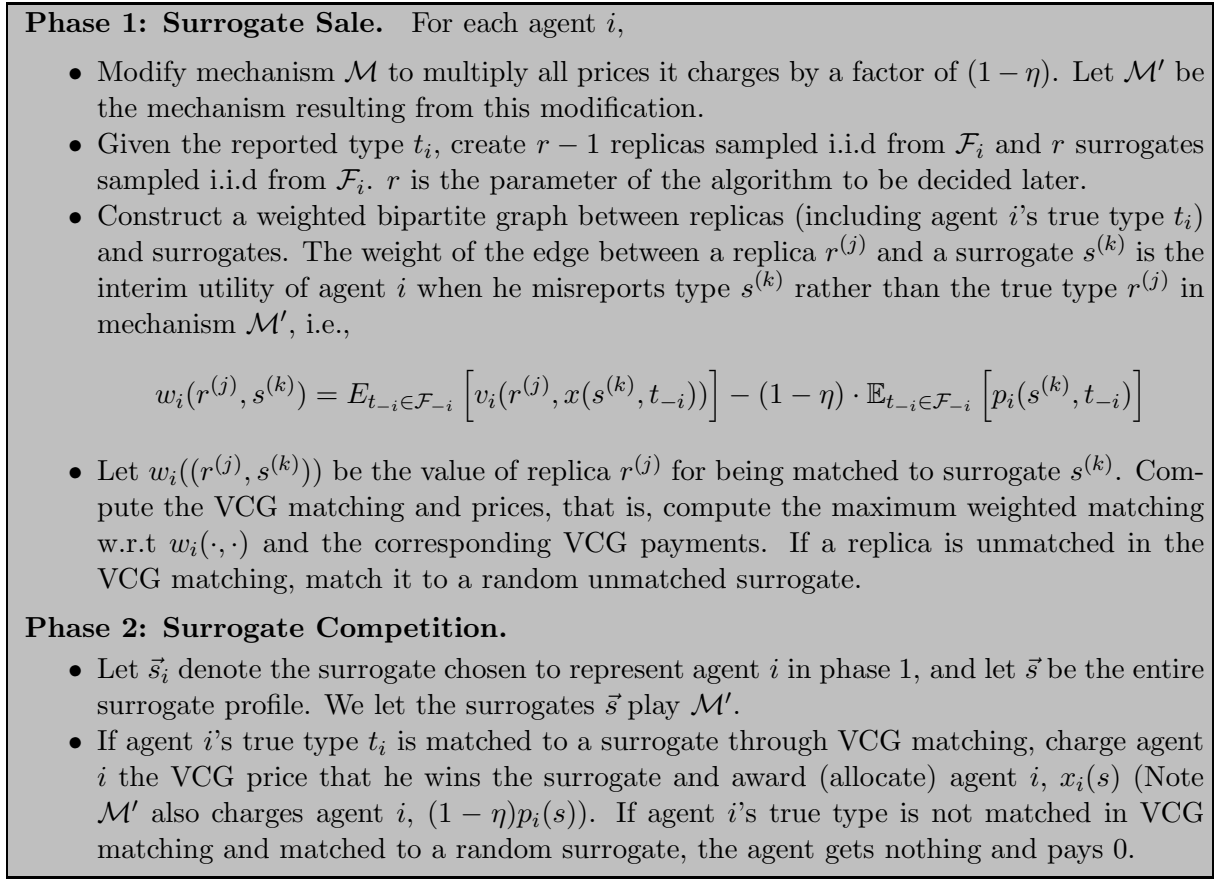

\fcolorbox{black}{gray!50!white}{
\parbox{0.95\textwidth}{
{\bf Phase 1: Surrogate Sale. } For each agent $i$,

\begin{itemize}
\item Modify mechanism $\M$ to multiply all prices it charges by a factor of $(1-\eta)$. Let $\M'$ be the mechanism resulting from this modification.

\item Given the reported type $t_i$, create $r-1$ replicas sampled i.i.d from $\F_i$ and $r$ surrogates sampled i.i.d from $\F_i$. $r$ is the parameter of the algorithm to be decided later.

\item Construct a weighted bipartite graph between replicas (including agent $i$'s true type $t_i$) and surrogates. The weight of the edge between a replica $r^{(j)}$ and a surrogate $s^{(k)}$ is the interim utility of agent $i$ when he misreports type $s^{(k)}$ rather than the true type $r^{(j)}$ in mechanism $\M'$, i.e.,
\begin{align*}
w_i(r^{(j)}, s^{(k)}) %
= E_{t_{-i}\in \F_{-i}}\left[v_i(r^{(j)}, x(s^{(k)}, t_{-i}))\right] - (1-\eta)\cdot \E_{t_{-i}\in \F_{-i}}\left[p_i(s^{(k)}, t_{-i})\right]
\end{align*}

\item Let $w_i((r^{(j)}, s^{(k)}))$ be the value of replica $r^{(j)}$ for being matched to surrogate $s^{(k)}$. Compute the VCG matching and prices, that is, compute the maximum weighted matching w.r.t $w_i(\cdot, \cdot)$ and the corresponding VCG payments. If a replica is unmatched in the VCG matching, match it to a random unmatched surrogate.
\end{itemize}

\noindent{\bf Phase 2: Surrogate Competition.}
\begin{itemize}
\item Let $\vec s_i$ denote the surrogate chosen to represent agent $i$ in phase 1, and let $\vec s$ be the entire surrogate profile. We let the surrogates $\vec s$ play $\M'$.

\item If agent $i$'s true type $t_i$ is matched to a surrogate through VCG matching, charge agent $i$ the VCG price that he wins the surrogate and award (allocate) agent $i$, $x_i(s)$ (Note $\M'$ also charges agent $i$, $(1-\eta)p_i(s)$). If agent $i$'s true type is not matched in VCG matching and matched to a random surrogate, the agent gets nothing and pays $0$. 
\end{itemize}
}}
\caption{Replica-Surrogate Matching Mechanism.}
\label{transform:replica-surrogate-matching}
\end{figure*}

\section{Omitted Definitions}\label{app:omitted-definitions}

\begin{definition}[Individual Rationality]\label{def:ex-interim-IR}
	A BIC/$\eps$-BIC mechanism $\M$ satisfies interim individual rationality (interim IR) iff for all $i, v_i$:
	\begin{eqnarray*}
		\E_{t_{-i}\sim \F_{-i}} [u_i(t_i, \M(t))] \geq 0
	\end{eqnarray*}
	This becomes ex-post individual rationality (ex-post IR) iff for all $i, t_i, t_{-i}, u_i(t_i, \M(t)) \geq 0$ with probability 1, over the randomness of the mechanism.
\end{definition}

\begin{definition}[Interdependent private type]\label{def:interdependent-type}
	Each agent $i\in [n]$ has a private signal $s_i$, which captures her private information and the type of every agent $t_i$ depends on the entire signal profile, $s = (s_1, \cdots, s_n)$.
\end{definition}

\section{Omitted Proofs}\label{app:omitted-proofs}
\subsection{Proof of Claim~\ref{claim:single-agent-uniform-claim-1}}\label{app:single-agent-uniform-claim-1}
\begin{proof}
First, in Step 1, since we only rotate the allocation and payment of nodes in $\C$, the total weight of the edges from nodes in $\T\backslash \C$ to nodes in $\C$ remains the same. Second, each node in $\C$ achieves a utility no worse than before, so that the weight of each outgoing edge from nodes in $\C$ to nodes in $\T\backslash\C$ will not increase. Third, since $\C$ is the shortest cycle, there are no other edges among nodes in $\C$ in addition to edges in $\C$, which implies we cannot create new edges among nodes in $\C$ by this rotation. It follows that this rotation
decreases the total weights of graph $G$ by the weights of $\C$. Finally, the expected revenue achieved by types $t^{(1)},\cdots,t^{(l)}$ is still the same, since Step 1 only rotates the allocation and payment rules, and the probability of each type is the same. Combining the fact that each node gets a weakly preferred outcome, the social welfare does not decrease.	
\end{proof}

\subsection{Proof of Claim~\ref{claim:single-agent-uniform-claim-2}}\label{app:single-agent-uniform-claim-2}
\begin{proof}
In Step 2, we first prove that it can only create new edges with zero weight. A new edge created by Step 2 can only point to a node $\bar{t}\in S_t$. We show by contradiction, suppose we create a positive weight edge from $\hat{t}$ to $\bar{t}\in S_t$, then $u(\hat{t}, \M'(\bar{t})) - u(\hat{t}, \M'(\hat{t})) > 0$ for the current updated mechanism $\M'$, we have 
\begin{align*}
u(\hat{t}, \M'(\hat{t})) &< u(\hat{t}, \M'(\bar{t})) < u(\hat{t}, \M'(\bar{t})) + \eps_t \\
&\leq u(\hat{t}, \M'(\bar{t})) + u(t'', \M'(t'')) - u(t'', \M'(\bar{t}))\\
&\leq  u(\hat{t}, \M'(\bar{t})) + u(\hat{t}, \M'(\hat{t})) - u(\hat{t}, \M'(\bar{t})) & (\text{By definition of } t'')\\
& = u(\hat{t}, \M'(\hat{t})),
\end{align*}
which proves our claim. Second, it is straightforward to verify that Step 2 doesn't decrease social welfare since we only decrease payment in Step 2. Finally, in Step 2, we reduce the weight of every positive-weight outgoing edge associated with $t$ by $\min\{\eps_t, \bar{\eps}_t\}$. This is because for any node $t'$, s.t. there is a positive-weight edge between $t$ and $t'$, $t'$ cannot be the ancestor of $t$, otherwise, there is already a cycle, which contradicts Step 1.
\end{proof}

\subsection{Proof of Theorem~\ref{thm:vanish-ic-outcome}}\label{app:vanish-ic-outcome}
\begin{proof}
	We construct the same weighted directed graph $G=(\T, E)$ as in the proof of  Theorem~\ref{thm:single-agent-uniform}. Again, the target is to reduce the total weight of $G$ to zero, which leads to a BIC mechanism. We denote $M_e$ as the menus and $|M_e| = C$, and we have for each type $t^{(i)}$, that there exists a menu $m_e \in M_e$, s.t. $\M(t^{(i)}) = m_e$. If $t^{(i)}$ and $t^{(j)}$ share a same menu, i.e., $\M(t^{(i)}) = \M(t^{(j)})$, there is an directed edge with weight zero from $t^{(i)}$ to $t^{(j)}$, and vice versa. We denote the distribution of each menu $m_e$ as,
	\begin{eqnarray*}
		g(m_e) = \sum_{t\in T: \M^\eps(t)=m_e}f(t).
	\end{eqnarray*}

	Since $\M$ is $\eps$-BIC, the weight of each edge is bounded by $\eps$. We still apply Step 1 and Step 2 in graph $G$ proposed in Theorem~\ref{thm:single-agent-uniform}, however, we count the revenue loss over menu space. 
	
	First, in Step 1, we only rotate the allocation and payment (menu) along the cycle, it will not change the allocation and payment of each menu. In addition, it will not the distribution of menus, $g(m_e)$ is preserved for each $m_e$. %
	
	In Step 2, consider a source node $t$, and let the corresponding menu be $m'_e$ (the output of the current mechanism with type $t$). Every type with $m'_e$
	is the ancestor of type $t$, when we decrease the payment of type $t$ by $\min\{\eps_t, \bar{\eps}_t\}$, the payment for each type $t'$ associated with menu $m'_e$ will be decreased by the same amount. 
	If there is a type $t''$ with a different menu $m''_e \neq m'_e$ and $t''$ is an ancestor of $t$, then all the types associated with menu $m''_e$ are the ancestors of $t$. Thus, in Step 2, the payment of the types with the same menu must be decreased by the same amount. Therefore, Step 2 only changes the payment of each menu by the same amount, and does not  change the distribution of each menu, i.e. $g(m_e)$ is the same for each $m_e\in M_e$.
	
	Moreover,  if there is an edge $(t^{(j)}, t^{(k)})$ with positive weight and if $t^{(j)}$ and $t$ share the same menu, then (1) $t^{(k)}$ must be in different menus, and (2) $t^{(k)}$ is not the ancestor of $t$, otherwise, there exists a cycle, which contains a positive-weight edge.  Therefore, in Step 2, if we decrease the payment of type $t$ by $\min\{\eps_t, \overline{\eps}_t\}$, we also reduce the weight of edge  $(t^{(j)}, t^{(k)})$ by $\min\{\eps_t, \overline{\eps}_t\}$. In other words, we reduce the regret of all the nodes in menu $m_e$ by $\min\{\eps_t, \overline{\eps}_t\}$.

	Since the weight of each edge is bounded by $\eps$, then we may decrease the expected payment at most $\eps$ to reduce all the regret of the nodes belonging to menu $m_e$. In total, the revenue loss is bounded by $C\eps$.
\end{proof}

\subsection{Proof of Theorem~\ref{thm:single-agent-general}}\label{app:single-agent-general}
\begin{proof}
We construct a weighted directed graph $G = (\T, E)$, different with the one in Theorem~\ref{thm:single-agent-uniform}. A directed edge $e = (t^{(j)}, t^{(k)})\in E$ is drawn from $t^{(j)}$ to $t^{(k)}$ when the outcome (allocation and payment) of $t^{(k)}$ is weakly preferred by true type $t^{(j)}$, i.e. $u(t^{(j)}, \M(t^{(k)})) \geq u(t^{(j)}, \M(t^{(j)}))$, %
and the weight of edge $e$ is 
\vspace{-5pt}
\begin{eqnarray*}
w(e)=f(t^{(j)})\cdot f(t^{(k)})\cdot \Big[u(t^{(j)}, \M(t^{(k)})) - u(t^{(j)}, \M(t^{(j)}))\Big]
\vspace{-5pt}
\end{eqnarray*}
It is straightforward to see that $\M$ is BIC iff the total weight of all edges in $G$ is zero.

We show the modified transformation for this setting in Fig.~\ref{transform:single-agent-general}. Firstly, it is trivial that our transformation preserves IR, since neither Step 1 nor Step 2 reduces utility.
Then we show this modified Step 1 will strictly decrease the total weights of the graph $G$ and has no negative effect on social welfare and revenue. 

First, we observe each type in $\C$ achieves utility no worse than before, by truthful reporting.
Then, the weight of each outgoing edge from a type in $\C$ to a type in $\T\backslash \C$ will not increase. 

Second, we claim the total weight of edges from any node (type) $t\in \T\backslash C$ to nodes (types) in $\C$ does not increase. To prove this, we assume $w(t, t^{(j)})\geq 0, \forall t^{(j)}\in \C$, i.e. there is a edge from $t$ to any $t^{(j)}\in\C$ in $G$. This is WLOG, because if there is no edge between $t$ to some $t^{(j)}\in C$, we can just add an edge from $t$ to $t^{(j)}$ with weight zero, and this does not change the total weight of the graph.
We denote the mechanism updated after one use of Step 1 as $\M'$, and denote the weight function $w'$ for the graph $G'$ that is constructed from $\M'$. Let $[\cdot]_+$ be the function $\max(\cdot, 0)$. The total weight from $t$ to $t^{(j)}\in \C$ according to the mechanism $\M'=(x', p')$ is
\begin{eqnarray*}
\vspace{-10pt}
&&\sum_{t^{(j)}\in\C} w'(t, t^{(j)})\\
&=& \sum_{j=1}^lf(t^{(j)}) f(t) \big[u(t, \M'(t^{(j)})) - u(t, \M'(t))\big]_{+}\\
&=& \sum_{j=1}^l f(t^{(j)}) f(t)\left[\frac{(f(t^{(j)}) - f(t^{(k)})) u(t, \M(t^{(j)})) + f(t^{(k)})\cdot u(t, \M(t^{(j+1)}))}{f(t^{(j)})} - u(t, \M(t))\right]_+\\
& & (\text{In the fractional rotation step, } \M'(t) = \M(t), \forall t \in \T\backslash \C) \\
&\leq& \sum_{j=1}^l f(t)\cdot \Big(\big(f(t^{(j)})) - f(t^{(k)})\big) \big[u(t, \M^\eps(t^{(j)})) - u(t, \M^\eps(t))\big]_+\\
&& + \sum_{j=1}^l f(t^{(k)})\big[u(t, \M(t^{(j+1)})) - u(t, \M(t))\big]_+\Big)\\
& & (\text{By rearranging the algebra and the fact that } [x + y]_+ \leq [x]_+ + [y]_+) \\
& = &\sum_{j=1}^l f(t^{(j)})\cdot f(t)\cdot \big[u(t, \M(t^{(j)}))- u(t, \M^\eps(t))\big]_+\\
& & (\text{By the fact that } \{t^{(1)}, \cdots, t^{(l)}\} \text{ forms a cycle and } t^{(l+1)} = t^{(1)})\\
& =& \sum_{t^{(j)}\in\C} w(t, t^{(j)})
\vspace{-10pt}
\end{eqnarray*}

\begin{figure*}[h]
\fcolorbox{black}{gray!50!white}{
\parbox{0.95\textwidth}{
{\bf Modified Step 1 (Fractional rotation step).}
Given a mechanism $\M=(x, p)$, find the shortest cycle $\C$ in $G$ that contains at least one edge with positive weight in $E$. Without loss of generality, we represent $\C = \left\{t^{(1)}, t^{(2)},\cdots, t^{(l)}\right\}$. Then we find the node $t^{(k)}, k\in[l]$, such that $f(t^{(k)}) = \min_{k\in[l]}f(t^{(k)})$. Next, we rotate the allocation and payment rules of types along $\C$ with fraction of $f(t^{(k)})/f(t^{(j)})$ for each type $t^{(j)}, j\in[l]$. %
Now we slightly abuse the notation of subscripts, s.t. $t^{(l+1)} = t^{(1)}$. Specifically, the allocation and payment rules for each $t^{(j)}$,
\begin{align*}
x'(t^{(j)}) &=\frac{\big[f(t^{(j)}) - f(t^{(k)})\big]x(t^{(j)}) + f(t^{(k)}) x(t^{(j+1)})}{f(t^{(j)})},\\
p'(t^{(j)}) &=\frac{\big[f(t^{(j)}) - f(t^{(k)})\big] p(t^{(j)}) + f(t^{(k)}) p(t^{(j+1)})}{f(t^{(j)})}.
\end{align*}

Then we update mechanism $\M$ to adopt allocation and payment rules $x', p'$ to form a new mechanism $\M'$ and reconstruct the graph $G$. If this has the effect
of removing all cycles  that contain at least one positive-weight-edge in $G$, then move to Step 2. Otherwise, we repeat Step 1.

{\bf Modified Step 2 (Payment reducing step). } Exactly the same as Step 2 in Theorem~\ref{thm:single-agent-uniform}.

}}
\caption{$\eps$-BIC to BIC transformation for single agent with general type distribution.}
\label{transform:single-agent-general}
\vspace{-10pt}
\end{figure*}

Thus, we prove our claim that the total weight of edges from any node (type) $t\in \T\backslash C$ to nodes (types) in $\C$ does not increase.

Third, by each use of modified Step 1, we remove one cycle and reduce the weight of edge $(t^{(i)}, t^{(i+1)})$ to zero, thus, we decrease the total weight at least by 
$f(t^{(k)})f(t^{(k+1)})(u(t^{(k)}, \M(t^{(k+1)}))-u(t^{(k)}, \M(t^{(k)}))$.

Finally, after one use of Step 1, the expected revenue achieved by types in $\C$ maintains, because
\begin{align*}
\sum_{j=1}^l f(t^{(j)})\cdot p'(t^{(j)}) &=  \sum_{j} (f(t^{(j)}) - f(t^{(k)}))\cdot p(t^{(j)}) + f(t^{(k)})\cdot p(t^{(j+1)})\\
&= \sum_j f(t^{(j)}) p(t^{(j)}) +  f{t^{(k)}}\sum_j   p(t^{(j)}) -  p(t^{(j+1)})\\
& = \sum_j f(t^{(j)}) p(t^{(j)}) & (\text{Because } t^{(l+1)} = t^{(1)})
\end{align*}

The modified Step 2 is the same as Step 2 in Fig.~\ref{transform:single-agent-uniform}. At each step 2, we decrease the total weight of the graph by at least $\min\{\eps_t, \overline{\eps}_t\}$. We count the revenue loss as follows, in each Step 2, if we decrease the payment of $t$ by $\min\{\eps_t, \overline{\eps}_t\}$, the expected revenue loss is bounded by
\begin{eqnarray*}
\sum_{j} f(t^{(j)}) \min\{\eps_t, \overline{\eps}_t\} \leq \min\{\eps_t, \overline{\eps}_t\}
\end{eqnarray*}
Since the weight of each edge is bounded by $\eps$, to reduce the weight of outgoing edges of $t$ to zero, we may decrease the expected revenue by $\eps$. Therefore, in total, the expected revenue loss is bounded by $m\eps$.
\end{proof}

\subsection{Proof of Theorem~\ref{thm:impossibility-EEIC}}\label{app:impossibility-EEIC}
\begin{proof}
We construct the type distribution and the $\eps$-EEIC mechanism similar to the one in Theorem~\ref{thm:lower-bound-revenue-loss}. We consider a single agent with $m$ types $\T=\{t^{(1)},\cdots,t^{(m)}\}$. The type distribution is $f(t^{(1)}) = \frac{1}{2} - \frac{\eps}{2m}$, $f(t^{(2)}) = \frac{\eps}{2m}$ and $f(t^{(j)})=\frac{1}{2(m-2)}, \forall j\geq 3$. The agent with type $t^{(1)}$ values outcome 1 at $\eps$ and the other outcomes at $0$. For any type $t^{(j)}, j\geq 2$, the agent with type $t^{(j)}$ values outcome $j-1$ at $m+(j-1)\eps$, outcome $j$ at $m+(j-1)\eps$, and the other outcomes at $0$. The mechanism we consider is: (1) if the agent reports type $t^{(1)}$, gives the outcome 1 to the agent and charges $\eps$. (2) if the agent reports $t^{(j)}, j\geq 2$, gives the outcome $j$ to the agent and charges $m+(j-1)\eps$. There is a $m$ regret to the agent for not misreporting type $t^{(1)}$ with true type $t^{(2)}$ and a regret $\eps$ for not reporting $t^{(j)}$ with true type $t^{(j+1)}$, for any $j\geq 2$. It is easy to verify that this mechanism is $\eps$-EEIC (the probability of type $t^{(2)}$ is small) and already maximizes social welfare. Thus, we can only change the payment to reduce the regret of each type. Following the same argument as in Theorem~\ref{thm:lower-bound-revenue-loss}, to reduce all the regret of the types, the revenue loss in total is at least
\begin{eqnarray*}
f(t^{(2)}) m + \sum_{j=3}^m f(t^{(j)})(m+(j-2)\eps) &=& \frac{\eps}{2} + \frac{1}{2(m-2)}\sum_{j=3}^m m+(j-2)\eps\\
&=& \frac{\eps}{2} + \frac{m}{2} + \frac{(m-1)\eps}{4} \geq \frac{m}{2}
\end{eqnarray*}
\end{proof}

\subsection{Proof of Theorem~\ref{thm:multi-agent-BIC}}\label{app:multi-agent-BIC}

The earlier proof approach for single agent case does not immediately extend to the multi-agent setting. However, since our target is a BIC mechanism, we can work with interim rules (see Definition~\ref{def:interim-rule}), and this 
provides an approach to the transformation. The interim rules  reduce the dimension of type space and separate the type of each agent. With this, we can construct a separate type graph for each agent, now based on the interim rules. 

To simplify the presentation, we define the {\em induced mechanism} for each agent $i$ of a mechanism $\M$ as follows.
\begin{definition}[Induced Mechanism]\label{def:interim-rule-induced-mechanism}
	For a mechanism $\M=(x, p)$, an induced mechanism $\widetilde{\M}_i = (X_i, P_i)$ is a pair of interum allocation rule $X_i: \T_i \rightarrow \Delta(\O)$ and interim payment rule $P_i: \T_i \rightarrow \R_{\geq 0}$. Denote the utility function $u_i(t_i, \widetilde{\M}_i(t_i)) = v_i(t_i, X_i(t_i)) - P_i(t_i)$.
\end{definition}

The following lemma shows that given an $\eps$-BIC/$\eps$-EEIC mechanism, then the induced mechanism for each agent is also $\eps$-BIC/$\eps$-EEIC. %
\begin{lemma}\label{lem:EEIC-decomposition}
	For a $\eps$-EEIC/$\eps$-BIC mechanism $\M$, any induced mechanism $\widetilde{\M}_i$ for each agent $i$ is $\eps$-EEIC/$\eps$-BIC.
\end{lemma}

\begin{proof}
	By $\eps$-BIC definition, each induced mechanism $\widetilde{\M}_i$ must be $\eps$-BIC, if the original mechanism $\M$ is $\eps$-BIC. Now, we turn to consider $\eps$-EEIC mechanism $\M$, for any induced mechanism $\widetilde{\M}_i$
	\begin{eqnarray*}
		&&\E_{t_i\sim \F_i}\left[\max_{t'_i\in\T_i} u_i(t_i,\widetilde{\M}_i(t'_i)) - u_i(t_i, \widetilde{\M}_i(t_i))\right]\\
		&=& \E_{t_i\sim \F_i}\left[\max_{t'_i\in\T_i} \E_{t_{-i}\sim \F_{-i}} \left[u_i(t_i, \M(t'_i; t_{-i})) - u_i(t_i, \M(t_i; t_{-i}))\right]\right]\\
		&\leq& \E_{t_i\sim \F_i}\left[\E_{t_{-i}\sim \F_{-i}} \left[\max_{t'_i\in\T_i} u_i(t_i, \M(t'_i; t_{-i})) - u_i(t_i, \M(t_i; t_{-i}))\right]\right] \\
		&& (\text{By Jenson's inequality and convexity of $\max$ function})\\
		&=& \E_{t\sim \F}\left[\max_{t'_i\in\T_i} u_i(t_i, \M(t'_i; t_{-i})) - u_i(t_i, \M(t_i; t_{-i}))\right]\\
		&& (\text{By independence of agents' types})\\
		&\leq& \eps.
	\end{eqnarray*}
\end{proof}

Given Lemma~\ref{lem:EEIC-decomposition}, we can construct a single type graph for each agent based on the induced mechanism and apply the same technique for each graph as the one in Theorem~\ref{thm:single-agent-general}.
The challenge will be to also handle feasibility of the resulting mechanism. We summarize these approaches in the following proof for Theorem~\ref{thm:multi-agent-BIC}.

\begin{proof}[Proof of Theorem~\ref{thm:multi-agent-BIC}]
Here, we focus on the $\eps$-BIC setting. The proof for $\eps$-EEIC with independent uniform type distribution is analogous.

We construct a graph $G_i=(\T_i, E_i)$ for each agent $i\in [n]$, such that there is a directed edge from $t_i^{(j)}$ to $t_{i}^{(k)}$ if and only if $u(t_i^{(j)}, \widetilde{\M^\eps}_i(t_{i}^{(k)})) \geq u(t_i^{(j)}, \widetilde{\M^\eps}_i(t_{i}^{(j)}))$ and the weight is 
\begin{eqnarray*}
w_i((t_i^{(j)}, t_i^{(k)})) = f_i(t^{j})\cdot f_i(t_i^{(k)}) \cdot \left(u(t_i^{(j)}, \widetilde{\M}_i(t_{i}^{(k)})) - u(t_i^{(j)}, \widetilde{\M}_i(t_{i}^{(j)})\right)
\end{eqnarray*}
Based on Lemma~\ref{lem:EEIC-decomposition}, each graph is constructed by an $\eps$-BIC induced mechanism $\widetilde{\M^\eps}_i$, we can apply the same constructive proof in Theorem~\ref{thm:single-agent-general} to reduce the total weight of each graph $G_i$ to be 0. An astute reader may have already realized that changing type graph $G_i$ may affect other graphs, since we probably change the distribution of the reported type of agent $i$. However, in our transformation, both Step 1 and Step 2 don't change the density probability of each type (we only change the interim allocation and payment for each type), therefore when we do transformation for one type graph $G_i$ of agent $i$, it has no effect on the interim rules of the other agents. 

Here, if the total weight of all graphs $G_i$ are all 0, it implies that any induced mechanism $\widetilde{\M^\eps}_i$ is IC. Therefore, we make the mechanism BIC. Similarly, the new mechanism after transformation achieves at least the same social welfare and the revenue loss of each graph $G_i$ is bounded by $m_i\eps$, Hence, the total revenue loss is bounded by $\sum_{i=1}^{n}m_i \eps = \sum_{i=1}^n |\T_i|\eps$.

What is left to show is that using modified steps 1 and 2 on each graph $G_i$ shown in Theorem~\ref{thm:single-agent-general} does not violate the feasibility of the mechanism. We only change the allocation of each type in modified Step 1 (Rotation step). Denote by $X_i$  the interim allocation for agent $i$ before one rotation step, and let $X'_i$ denote the updated interim allocation for agent $i$ after one rotation step. We then claim in the modified Step 1 in Theorem~\ref{thm:single-agent-general},
\begin{eqnarray}\label{eq:allocation-invariant-multi-agent}
\sum_{t_i\in \T_i} f_i(t_i) X_i(t_i) = \sum_{t_i\in \T_i} f_i(t_i) X'_i(t_i).
\end{eqnarray}
	
To prove this claim, WLOG, we consider a $l$ length cycle $\C = \{t_i^{(1)}, t_i^{(2)}, \cdots, t_i^{(l)}\}$ in modified Step 1. Let $k = \argmin_{j\in[l]} f_i(t^{(j)})$. We observe the interim allocation of the types in $\T_i\backslash \C$ don't change in modified Step 1, i.e., $\forall t_i \in \T_i\backslash \C, X_i(t_i) = X'_i(t_i)$. We slightly abuse the notation here, and let $t^{(l+1)} = t^{(1)}$. For the types in cycle $\C$, 
\begin{eqnarray*}
\sum_{j\in [l]} f_i(t^{(j)}) X'_i(t^{(j)}) &=& \sum_{j\in [l]} f(t^{(j)}) \cdot \frac{(f(t^{(j)}) - f(t^{(k)})) X_i(t^{(k)}) + f(t^{(k)})\cdot X_i(t^{(j+1)})}{f(t^{(j)})}\\
&=&  \sum_{j\in [l]} f_i(t^{(j)}) X_i(t^{(j)}),
\end{eqnarray*}
which validates the claim. Therefore, by Border's lemma~\citep{Border91}, the rotation step maintains the feasibility of the allocation.

\noindent{\bf Running time. } Suppose we have oracle access to the interim quantities of the original mechanism, we can build each $G_i$ in  $\mathtt{poly}(|\T_i|)$ time. Then, the running time for each type graph $G_i$ is $\mathtt{poly}(|\T_i|)$ following the same argument for single agent. In total the running time is $\mathtt{ploy}(\sum_i|\T_i|)$.  
\end{proof}

\subsection{Proof of Theorem~\ref{thm:lb-revenue-loss-multiple-agents}}\label{app:lb-revenue-loss-multiple-agents}
\begin{proof}
It is straightforward to construct an example such that the type graph of each agent induced by the interim rules is the same as the type graph constructed by the mechanism shown in Theorem~\ref{thm:lower-bound-revenue-loss}. For instance, agent $i$ values outcomes $\{o^{(1)}_i, \cdots, o^{(m_i)}_i\}$ in the same way as the one constructed in Theorem~\ref{thm:lower-bound-revenue-loss}. We assume the outcome $o^{(j)}_i$ are disjoint, for any $i$ and $j\in [m_i]$. Indeed, this is also $\eps$-DSIC mechanism. Thus, we show for this case, that the revenue loss must be at least $\Omega(\sum_{i}|\T_i|\eps)$, if we want to maintain the social welfare, following the same argument in Theorem~\ref{thm:lower-bound-revenue-loss}.
\end{proof}

\subsection{Proof of Theorem~\ref{thm:failure-interdependent-type}}\label{app:failure-interdependent-type}
\begin{proof}
	Consider a setting with two items $A$ and $B$ and two unit-demand agents $1$ and $2$. %
	The two agents share the same preference order on items. Moreover, agent 1 is informed about which is better, while agent 2 has no information. Agent 1 values the better item at $1+\eps$ and the other item
	at 1. Agent 2 values the better item at $2$ and the other item at 0.
	
	There exists an $\eps$-IC mechanism: ask agent 1 which item is better, and give this item to agent 2 for a price of 2 and give agent 1 the other item for a price of 1.  The total welfare and revenue is 3 if agent 1 reports truthfully. Bidder 1 can get $\eps$ more utility by misreporting, in which case it will get the better item for the same price. From this, we can confirm that this is an $\eps$-IC mechanism.
	
	For any IC mechanism, by weak monotonicity, we have $v_A(x(A)) - v_A(x(B)) \geq v_B(x(A)) - v_B(x(B))$, where $v_A$ be the type that the better item is $A$, and similarly for $v_B$. $x(A)$ is the allocation if agent 1 reports $A$ the better item and similarly for $x(B)$.
	This means that when agent 1 reporting $A$ rather than $B$, either agent 1 is assigned item $A$ with weakly higher probability, or agent 1 is assigned item $B$ with weakly less probability.  We only consider the former case, and the latter one holds analogously.
	In the former case, we have either:
	
	(1) agent 1 is getting at least half of $A$ when reporting $A$, and the total revenue and social welfare are each at most $0.5\times 2 + 0.5\times(1+\epsilon) + 1 = 2.5+\epsilon/2$, or
	
	(2) agent 1 is getting at most half of $A$ when reporting $B$, and the total revenue and social welfare are each at most $2 + 0.5 = 2.5$.
	
	Either way, we will definitely lose at least $0.5 - \eps/2$ for revenue and social welfare when making the $\eps$-IC mechanism above BIC.
\end{proof}

\subsection{Proof of Theorem~\ref{thm:failure-dsic-target}}\label{app:failure-dsic-target}
\begin{proof}
The construction of this $\eps$-BIC mechanism is strictly generalized by the mechanism in \citep{Yao17}. Consider a 2-agent, 2-item auction, each agent $i$ values item $j$, $t_{ij}$. $t_{ij}$ is i.i.d sampled from a uniform distribution over set $\{1, 2\}$, i.e. $\PP(t_{ij} = 1) = \PP(t_{ij}=2)  = 0.5$. The $\eps$-BIC mechanism is shown as below,

\smallskip
\fcolorbox{black}{gray!50!white}{
\centering
\parbox{0.9\textwidth}{
If $t_2=(1,1)$, give both items to agent 1 for a price of $3$. 

If $t_2=(1,2)$ and $t_1=(1,2)$, give both items randomly to agent 1 or 2 for a price of $1.5$.

If $t_1=(2,1)$ and $t_2=(1,2)$, give item 1 to agent 1 and give item 2 to agent 2, with a price of $2$ for each.

If $t_2=(1,2)$ and $t_1=(2,2)$, give both items to agent 1 for a price $3.75+\eps$.

If $t_1 = t_2=(2,2)$, give both items randomly to agent 1 or agent 2 for a price $2$.

For other cases, we get the mechanism by the symmetries of items and agents.
}}
\smallskip

It is straightforward to verify that this is an $\eps$-BIC mechanism and the expected revenue is $3.1875 + \eps/16$. However, \citet{Yao17} characterizes that optimal DSIC mechanism achieves expected $3.125$. This conclude the proof.
\end{proof}

\section{Omitted Details of Applications}\label{app:omitted-application}
In this section, we give a brief introduction to LP-based AMD and RegretNet AMD. 
\subsection{LP-based Approach}
The LP-based approach considered in this paper is initiated by~\citep{ConitzerS02}. We consider $n$ agents with type distribution $\F$ defined on $\T$. For each type profile $t\in \T$ and each outcome $o_k\in O$, we define $x^k(t)$ as the probability of choosing $o_k$ when the reported types are $t$ and $p_i(t)$ as the expected payment of agent $i$ when the reported types are $t$. $x^k(t)$ and $p_i(t)$ are both decision variables.

Then we can formulate the mechanism design problem as the following linear programming,
\begin{eqnarray*}
&&\max_{x, p} (1-\lambda)\E_{t\sim \F} \left[\sum_i p_i(t)\right] + \lambda \E_{t\sim \F}\left[\sum_{k:o_k \in O} x^k(t) \sum_i v_i(t_i, o_k)\right]\\
&s.t.& \E_{t_{-i}} \left[\sum_{k:o_k \in O} x^k(t_i, t_{-i}) v_i(t_i, o_k) - p_i(t_i, t_{-i})\right] \geq \E_{t_{-i}} \left[\sum_{k:o_k \in O} x^k(t'_i, t_{-i}) v_i(t_i, o_k) - p_i(t'_i, t_{-i})\right], \forall i, t_i, t'_i\\
&& \E_{t_{-i}} \left[\sum_{k:o_k \in O} x^k(t) v_i(t_i, o_k) - p_i(t)\right]\geq 0, \forall i, t
\end{eqnarray*}
where the first constraint is for BIC and the second is for interim-IR. In this case, the type space $\T$ is discrete, thus the expectation can be explicitly represented as the linear function with decision variables.

\subsection{RegretNet Approach}
RegretNet~\citep{Duetting19} is a generic data-driven, deep learning framework for multi-dimensional mechanism design. We only briefly introduce the RegretNet framework here and refer the readers to~\citep{Duetting19} for more details.

RegretNet uses a deep neural network parameterized by $w\in\R^d$ to model the mechanism $\M$, as well as the valuation (through allocation function $x^w: \T\rightarrow \Delta(O)$) and payment functions: $v_i^w: \T_i\times \Delta(O) \rightarrow \R_{\geq 0}$ and $p_i^w: \T \rightarrow \R_{\geq 0}$. Denote utility function as,
$$u_i^w(t_i, \hat{t}) = v_i(t_i, x^w(\hat{t})) - p_i^w(\hat{t}).$$ 
RegretNet is trained on a training data set $\mathcal{S}$ of $S$ type profiles i.i.d\ sampled from $\F$ to maximize the empirical revenue subject to the empirical regret being zero for all agents:
\begin{eqnarray*}
&&\max_{w\in \R^d} \frac{1-\lambda}{S}\sum_{t\in \mathcal{S}} \sum_{i=1}^n p_i^w(t) + \frac{\lambda}{S}\sum_{t\in \mathcal{S}} \sum_{i=1}^n v_i^w(t_i, x(t))\\
&s.t.& \frac{1}{S}\sum_{t\in \mathcal{S}} \left[\max_{t'_i\in\T_i} u_i^w(t_i, (t'_i, t_{-i})) - u_i^w(t_i, t) \right] = 0, \forall i
\end{eqnarray*}
The objective is the empirical version of learning target in~\ref{eq:learning-target-amd}.
The constraint is for EEIC requirement and IR is hard coded in RegretNet to be guaranteed. Let $\mathcal{H}$ be the functional class modeled by RegretNet through parameters $w$. In this paper, we assume there exists an PAC learning algorithm that can produce a RegretNet to model an $\eps$-EEIC mechanism $\M\in \mathcal{H}$ defined on $\F$, such that
\begin{eqnarray*}
\mu_\lambda(\M', \F) \geq \sup_{\hat{\M}\in \mathcal{H}}\mu_\lambda(\hat{\M}, \F)  - (1-\lambda)\sum_{i=1}^n \vert \T_i\vert \eps - \eps,
\end{eqnarray*}
holds with probability at least $1-\delta$, by observing $S=S(\eps, \delta)$ i.i.d\ samples from $\F$.

\end{document}